\documentclass[a4paper,11pt]{article}
\usepackage{amsmath,amsthm}
\usepackage{amsfonts}
\usepackage{amssymb}
\usepackage{makeidx}
\usepackage{graphicx}
\usepackage{xcolor}
\usepackage{lmodern}
\usepackage{authblk}
\usepackage{bbm}
\usepackage{bm}
\usepackage{physics}
\usepackage{braket}
\usepackage[left=2cm,right=2cm,top=2cm,bottom=2cm]{geometry}
\usepackage{tikz}
\usepackage{verbatim}

\newtheorem{theorem}{Theorem}
\newtheorem{lemma}{Lemma}

\newtheorem{proposition}{Proposition}
\newtheorem{remark}{Remark}
\newtheorem{definition}{Definition}

\makeatletter
\newtheorem*{rep@theorem}{\rep@title}
\newcommand{\newreptheorem}[2]{%
\newenvironment{rep#1}[1]{%
 \def\rep@title{#2 \ref{##1}}%
 \begin{rep@theorem}}%
 {\end{rep@theorem}}}
\makeatother
\newreptheorem{theorem}{Theorem}

\newcommand{\bignumber}{\mbox{\normalfont\Large\bfseries $\dfrac{2}{|\delta V|}$}}
\newcommand{\varEpsilon}{\mbox{\normalfont\Large\bfseries $\varepsilon$}}

\usepackage{color}

\usepackage{ulem}

\title{
A comfortable graph structure for Grover walk 
}
\author{
Yusuke Higuchi\thanks{Department of Mathematics,
Gakushuin University,
Tokyo 171-8588, Japan},
Mohamed Sabri\thanks{Graduate School of Information Sciences, Tohoku University, Aoba, Sendai 980-8579, Japan
} and  
 Etsuo Segawa\thanks{Graduate School of Environment and Information Sciences, Yokohama National University, Hodogaya, Yokohama 240-8501, Japan 
 } 

 }
 \date{}

\begin{document}
\maketitle

 \noindent {\bf Abstract.} 
We consider a Grover walk model on a finite internal graph,
which is connected with a finite number of
semi-infinite length paths and receives the alternative inflows along
these paths at each time step.
After the long time scale, we know that the behavior of such a Grover walk
should be stable, that is, this model has a stationary state.
In this paper our objectives are to give
some characterization upon the scattering of the stationary state
on the surface of the internal graph
and upon the energy of this state in the interior. 
For the scattering, we concretely give a scattering matrix,
whose form is changed depending on whether the internal graph is bipartite or not. On the other hand, we introduce a comfortability function of a graph
for the quantum walk,
which shows how many quantum walkers can stay in the interior,
and we succeed in showing
the comfortability of the walker in terms of combinatorial
properties of the internal graph. \\
%
%


\section{Introduction}
As discrete-time quantum walks on graphs are being studied, many interesting behaviors of quantum walks, which cannot be observed in the classical random walk setting, have become apparent: the accomplishment of quantum speed up in quantum algorithm, linear spreading, localization, periodicity, pseudo-perfect state transfer and so on (see \cite{Konno, Manouchehri:Wang,Portugal} and  references therein). 
Among them it is revealed that behaviors of quantum walks are closely related to geometric features of graphs; for examples, cycle geometry of graphs gives the localization~\cite{HKSS,HS2}, 
a three-edge-coloring induces an  eigenfunction of some quantum walks~\cite{MaresNovotonyJex}, 
and the rotation systems and 1-factorizations of graphs are reflected on the mixing time of quantum walks~\cite{GodsilZhang}. 
For a family of graphs with fractal structure, the spectral analysis reflects on its property for the hopping matrix~\cite{DomanyEtAl} and the continuous-time quantum walk~\cite{Boe} are studied. 
In \cite{FG}, the scattering for a tree with two semi-infinite lines appended is discussed under the time evolution of continuous-time quantum walk to investigate the quantum exponentially speed up of the hitting time to the marked vertex.  
In this paper, we observe what property of general graphs are extracted from the structure of the stationary state of discrete-time quantum walks. The stationary state, which is a kind of the long time behavior of quantum walk, has been discussed in \cite{Feldman:Hillery:2005,Feldman:Hillery:2007,Higuchi:Segawa}.


Let us explain our setting. 
For a connected and locally finite graph $G=(V,E)$, which may be infinite, let us define the set of symmetric arcs $A$ induced by $E$ as follows: for any undirected edge $e \in E$ with end vertices $u,v \in V$, the induced arcs are $a$ and $\overline{a}$, which are arcs from $u$ to $v$ and $v$ to $u$ along the edge $e \in E$, respectively. The terminus and origin vertices of $a \in A$ are denoted by $t(a), \, o(a) \in V$, respectively.
The total space of quantum walk treated here is denoted by $\mathbb{C}^{A}$, which is a vector space whose standard basis are labeled by $A$. In other words, $\mathbb{C}^{A}\cong \{\psi:A\to \mathbb{C}\}$ is spanned by $\left\lbrace \delta_{a} : a \in A \right\rbrace$, where $\delta_{a}(\cdot)$ is the delta function such that
$$
\delta_{a}(a')= \begin{cases}
1 & \text{if } a = a',\\
0 & \text{otherwise}.
\end{cases}
$$
The time evolution operator of a quantum walk $U:\mathbb{C}^{A}\to \mathbb{C}^A$ is defined by
$$U=SC.$$
Here $S \delta_{a} = \delta_{\overline{a}}$ and $C= \bigoplus_{v \in V} C_{v}$, where $C_{v}$ is a local coin operator assigned at vertex $v$ which is a $\text{deg}(v)$-dimensional unitary matrix on $\text{span} \left\lbrace \delta_{a} : t(a) =v \right\rbrace$.
Note that since the time evolution operator $U$ is unitary in the sense that $UU^*=I$, then for any $\psi\in \mathbb{C}^A$, the $l^2$ norm is preserved, that is, $||U\psi||^2=||\psi||^2=\sum_{a\in A}|\psi(a)|^2$ for the standard inner product.   
Let $\psi_n$ be the $n$-th iteration of $U$ with the initial state $\psi_0$; that is, $\psi_{n+1}=U\psi_n$. We call  
$\psi_\infty(a):=\lim_{n\to\infty}\psi_n(a)$ for any $a\in A$, if it exists, the stationary state.
Another type of stationarity of quantum walks, the stationary {\it measure}, is discussed in \cite{KonnoTakei},
whose form is as follows:  
$\mu(\psi_{n+1})(u):=\sum_{t(a)=u}|\psi_{n+1}(a)|^2=\sum_{t(a)=u}|\psi_{n}(a)|^2$ for any time step $n=0,1,2,\dots$ and vertex $u\in V$. 

We are interested in how the state converges to the stationary state as a fixed point of a dynamical system. 
In general, the state of the walker in a quantum walk on a finite graph $G_0=(V_0,A_0)$ does not necessarily converge because the time evolution operator is unitary and hence the absolute values of the eigenvalues of the operator are equal to $1$. 
Then we connect semi-infinite paths, say tails, to a subset of vertices of $\delta V\subset V_0$ in the finite graph $G_0$, which is called the surface. Moreover the $l^\infty$-initial state $\Phi_0$ is set so that a quantum walker along the tails come into the internal and a quantum walker goes out from that at every time step. See (\ref{eq:initial_state}) for the detailed definition of the initial state and also Fig.~\ref{fig:evolution}. 
It is expected that the outflow balances the inflow in the long time limit and this quantum walk on this graph goes to a steady state.
Indeed the convergence of states, in this sense, is shown in \cite{Feldman:Hillery:2005, Feldman:Hillery:2007, Higuchi:Segawa}. 
In this setup, the state space of the walk is described by $l^{\infty}$-space. 
This quantum walk always converges to a stationary state whenever we choose vertices of the internal, to which the tails are connected. 

In \cite{Higuchi:Sabri:Segawa,Higuchi:Segawa}, the Grover walk whose quantum coins are expressed by the Grover matrices $\left\lbrace \mathrm{Gr}(\text{deg}(u) \right\rbrace_{u \in V}$ was applied with the constant inflow. Then the scattering on the surface $\delta V$, which is a set of vertices connected to tails, in the long time limit recovers the local scattering at each vertex; that is, the scattering matrix is $\mathrm{Gr}(r)$, where $r$ is the number of tails. 
This implies that we can obtain the global scattering by only some information on the surface, especially the number $r$, whereas this also implies that we cannot obtain any information on the structure of the internal graph $G_{0}$. 

Our motivation is to investigate how the geometrical structure of the internal graph affects the behavior of quantum walkers in terms of their stationary states. For this purpose, we must change something in our setting; for example, we replace a constant flow, which was studied in \cite{Higuchi:Sabri:Segawa,Higuchi:Segawa}, with some oscillated one.
In this study, as a first trial, we focus on the simplest oscillation of the inflow with alternating signs which is a coarse graining of alternating current input. 
Our interest is to find a graph geometry from the scattering on the surface and also the energy in the interior in the long time limit. 
See Figure~\ref{fig:evolution}. 
\begin{figure}[htb]
    \centering
    \includegraphics[width=12cm]{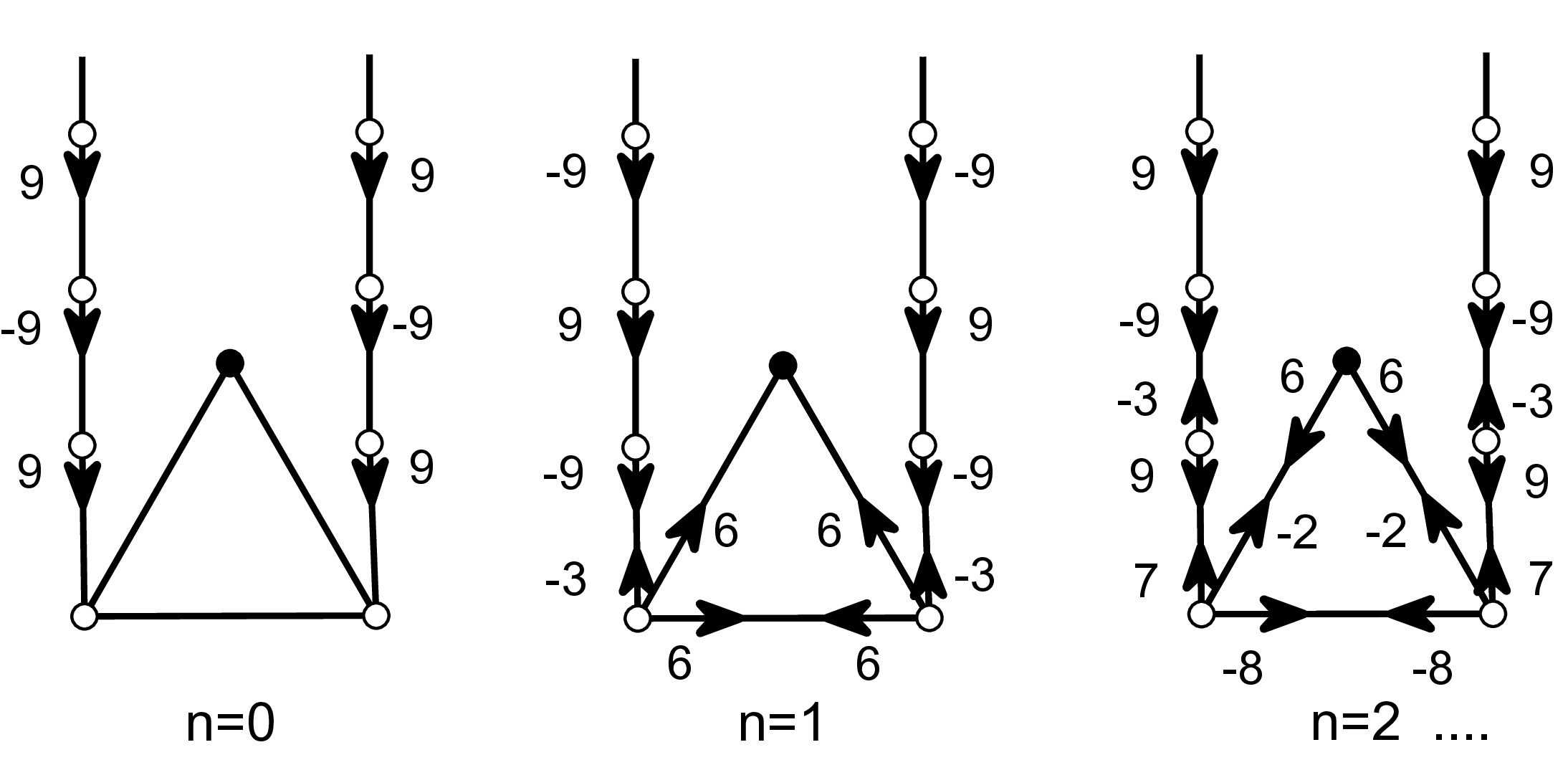}
    \caption{The time evolution of the quantum walk with the alternative inflow: The surface of the internal graph is the set of vertices in $K_3$ connected to the tails. At time $n=0$, we set the initial state $\Phi_0$ with alternating signs on the tails represented by $\boldsymbol{\alpha}=[\alpha_1,\alpha_2]^{\top}=[9,9]^\top$. The time evolution is determined by the Grover matrices assigned at each vertex: for each vertex $u$ and each time step, the transmitting weight is $2/\deg(u)$ while the reflection weight is $2/\deg(u)-1$. Then on the tails, the dynamics is free because $\deg(u)=2$ while in the internal graph, the time evolution is not trivial. For example, see the middle of the figure, at time $n=1$ and at the vertex connecting to the left tail, say $u_*$, a quantum walker who came from the tail with the amplitude ``$9$" at the previous time is transmitted to the internal graph with the amplitude $9\times 2/\deg(u_*)=9\times 2/3=6$, while reflected to the tails with the amplitude $9\times (2/\deg(u_*)-1)=-3$. Because of the setting of the initial state and free dynamics on the tails, the internal graph receives the inflow from the tails and also radiates the outflow to the tails at every time step.} 
    \label{fig:evolution}
\end{figure}

The scattering on the surface gives how the inflows map to the outflows on the surface. 
This answer is given in following.


\begin{theorem}
\label{main1}
Let $\bm{\alpha}$ and $\bm{\beta}$ represent the inflow and outflow of the stationary state on the surface $\delta V$. 
Their details can be seen in Definition \ref{def:scattering}. 
Then we have 
\[ \bm{\beta}=\sigma \bm{\alpha}, \]
where the scattering matrix $\sigma$ is unitary and described by 
\[ \sigma = 
\begin{cases}
I & \text{: $G_{0}$ is non-bipartite,}\\
\tau & \text{: $G_{0}$ is bipartite.}
\end{cases} \]
Here $I$ is the identity matrix and $\tau$ is described as in Section~3. 
\end{theorem}
\noindent From this theorem, if the interior is non-bipartite, the scattering is the perfect reflection, while if the interior is bipartite, the scattering is described by the Grover matrix.
Then in contrast with results in \cite{Higuchi:Sabri:Segawa,Higuchi:Segawa}, we succeed in extracting a structure of the internal graph from the scattering on the surface; that is, the bipartiteness. 

Next it is natural to ask what happens to the interior. To answer this question, we introduce the idea of comfortability. The comfortability is the function of the interior, and gives how quantum walkers accumulate in the internal graph in the long time limit, which is the energy stored the interior. 
The notion of the comfortability for quantum walk is an analogy of the energy in the electric circuit~\cite{DS}. 
The detailed definition of the comfortability 
$\mathcal{E}_{QW}(G_0)$ is described in Definition~\ref{def:comfortability}.
We obtain that when only two tails are connected, the comfortability of the graph can be expressed in terms of the geometric information of the graph as follows. 
\begin{theorem}
\label{thm:comfortability}
Assume the number of tails is $2$, and the inflow $\bm{\alpha}=[\alpha_1,\alpha_2]^\top=[1,0]^\top$. Then the comfortability $\varEpsilon_{QW}(G_{0})$ of the quantum walk is given by
$$
\varEpsilon_{QW}(G_{0})=
\begin{cases}
\left(\chi_{2}(G_{0})/\chi_{1}(G_{0})+ |E_{0}|\right)/4  & \text{ : } G_{0} \text{ is bipartite,}\\
\iota_{2}(G_{0})/\iota_{1}(G_{0}) & \text{ : } G_{0} \text{ is non-bipartite.}
\end{cases}
$$
Here $|E_0|$ is the number of edges of $G_0$,  $\chi_1(G_0)$ is the number of the spanning trees of $G_0$ and $\chi_2(G_0)$ is the number of the spanning forests of $G_0$ with two components in which one contains $u_1$ and the other contains $u_n$. The geometric quantities of $\iota_1(G_0)$ and $\iota_2(G_0)$ are defined in Section~3. 
\end{theorem}
Under the condition where the outflow balances the inflow whose amplitude is constant, we can control the amount of energy stored in the interior by tuning the oscillation of the inflow according to the geometric structure. 
We illustrate the ranking of the comfortability among all of the graphs with four vertices in Section~\ref{Example}. 
To show the above theorem, we first define a (pseudo-)current function~\cite{Higuchi:Sabri:Segawa, Higuchi:Segawa} on a (non-)bipartite graph and we show that these functions satisfy the (pseudo-)Kirchhoff laws in Theorems~\ref{main2} and \ref{thm:pKL}, respectively. Secondly, we obtain a potential function with respect to these (pseudo-)current functions, in terms of (signless-)Laplacian matrix in \cite{Biggs, Bollobas} and Theorem~\ref{signlesslaplacian}, respectively. Thirdly, using these expressions of the potential functions, we obtain expressions for the comfortability in terms of the (non-)oriented incident matrix and the (signless-)Laplacian matrix, when the graph is (non-)bipartite in Propositions~\ref{prop} and \ref{prop:2}. 
Finally, using the speciality of the setting, we obtain the expression of the comfortability using the graph factors induced quantum walks defined by Definition~\ref{def:subgraphs}. 

We arrange the rest of this paper in the following way. In section 2, we introduce the setting, notations and some previously known results in our model. Moreover, we define the scattering matrix of the walk and a comfortability function of the graph. In section 3, we introduce our main results of the paper: Theorems~\ref{main1} and \ref{thm:comfortability} which tells information on the scattering matrix and the comfortability function, respectively. These theorems depends on the bipartiteness of graphs. Furthermore, we note that the theorem on the comfortability gives a method to compute the comfortability of the graph in terms of the geometric information of the graph. In section 4, we give an example to illustrate the comfortability of the graphs with 4 vertices and we show the comparison of the comfortabilities. In section 5, we give the proofs of our main theorems, Theorems~\ref{main1} and \ref{thm:comfortability}. 


\section{Setting}
\label{Setting}
Let $G_{0}=(V_{0},E_{0})$ be a finite connected graph and $A_{0}$ be the symmetric arc set induced by $E_{0}$. We choose the boundary(surface) of $G_{0}$, $\emptyset \neq \delta V \subset V_{0}$ with $\delta V = \{ v_{1}^{(0)},...,v_{r}^{(0)}\}$, where $v_i^{(0)}\neq v_j^{(0)}$ if and only if $i \neq j$. 
Let $\left\lbrace \mathbb{P}_{j}:j=1,...,r \right\rbrace$ be the set of semi-infinite length paths called the tails each end vertex of which is denoted by $v_{j}^{(0)}$, and each of which is connected to the finite graph $G_{0}$ such that $V(\mathbb{P}_{j})=\left\lbrace v_{j}^{(0)} \sim v_{j}^{(1)} \sim v_{j}^{(2)} \sim ... \right\rbrace$.  Here $u \sim v$ means that the vertices $u$ and $v$ are adjacent. We denote the constructed graph by $\tilde{G}=(\tilde{V},\tilde{E})$. We also denote the arc set induced by $E_{0}$ and $\tilde{E}$ by $A_{0}$ and $\tilde{A}$ respectively. For any arc $a =(u,v) \in \tilde{A}$, we write $\overline{a} = (v,u)$, $o(a)=u$ and $t(a)=v$.
Remark that $o(\overline{a})=t(a)$ and $t(\overline{a})=o(a)$. 

For a discrete set $\Omega$, let us define $\mathbb{C}^{\Omega}$ by the vector space whose basis are labeled by each element of $\Omega$.
The total state space associated with the quantum walk treated here is $\mathbb{C}^{\tilde{A}}$. 
We define the time evolution operator $W$ on $\mathbb{C}^{\tilde{A}}$ in the matrix form by 
\[ (W)_{a,b}= \begin{cases} \left(\dfrac{2}{\text{deg}(o(a))}-\delta_{a \overline{b}} \right) & \text{ if } \, o(a)=t(b),\\
0 & \text{otherwise},
\end{cases} \]
which is so called the Grover walk. 
Note that the walk becomes ``free" on the tails; that is, 
\[ (W)_{a,b}= \begin{cases} 1 & \text{: $o(a)=t(b)$, $a\neq \overline{b}$,} \\
0 & \text{: otherwise}
\end{cases} \]
for any $o(a)\notin V_0$. 
We set the $l^\infty$-initial state by using a complex value $z$ with $|z|=1$:  
\begin{equation}\label{eq:initial_state}
\Phi_{0}(a)=
\begin{cases}
z^{-\text{dist}(v_{j}^{(0)},\;t(a))}\alpha_{j} & \text{if } o(a) = v_{j}^{(s+1)}, \, t(a) = v_{j}^{(s)}, \, s=0,1,2,..., \, j = 1, 2, ..., r,\\\\
0 & \text{otherwise}.
\end{cases}
\end{equation}
Here $\alpha_j\in \mathbb{C}$ $(j=1,\dots,r)$ and $\mathrm{dist}(u,v)$ is the shortest length of paths in $\tilde{G}$ between vertices $u$ and $v$. 
Then quantum walkers inflows into the internal graph $G_0$ at every time step $n$ from the tails. On the other hand, a quantum walker outflows towards the tails from the internal graph.

Let $\Phi_n\in \mathbb{C}^{\tilde{A}}$ be the $n$-th iteration of the quantum walk such that $\Phi_{n+1}=W\Phi_n$. Since the inflow oscillates with respect to the time step, the total state does not converge in the long time limit. We put $\Phi_n':=z^{n}\Phi_n$, which satisfies $\Phi_{n+1}'=zW\Phi_n'$. 
The convergence of $\Phi_{n}'$ is ensured by \cite{Higuchi:Segawa} as follows.  
\begin{theorem}[\cite{Higuchi:Segawa}]
 $\Phi_{\infty}' := \lim \limits_{n \rightarrow \infty}\Phi_{n}'$ exists; that is, $W\Phi_\infty'=z^{-1}\Phi_\infty'$.
\end{theorem}
Let us focus on the dynamics restricted to the internal graph.  
To this end, we define the boundary operator 
of $A_{0}$, $\chi: \mathbb{C}^{\tilde{A}} \rightarrow \mathbb{C}^{A_{0}}$ by,
$$(\chi \tilde{f})(a)=\tilde{f}(a), a \in A_{0}$$
for any $\tilde{f}\in \mathbb{C}^{\tilde{A}}$. 
The adjoint $\chi^{*} : \mathbb{C}^{A_{0}} \rightarrow \mathbb{C}^{\tilde{A}}$ is described by
$$(\chi^{*} f)(a) =
\begin{cases}
f(a) & \text{if} \, a \in A_{0},\\
0 & \text{otherwise}
\end{cases}
$$
for any $f\in \mathbb{C}^{A_0}$. 
Remark that $\chi \chi^{*} : \mathbb{C}^{A_{0}} \rightarrow \mathbb{C}^{A_{0}}$ is the identity operator on $\mathbb{C}^{A_{0}}$ and $\chi^{*} \chi : \mathbb{C}^{\tilde{A}} \rightarrow \mathbb{C}^{\tilde{A}}$ is the projection operator on $\mathbb{C}^{\tilde{A}}$ with respect to $A_{0}$. 
Putting $\chi\Phi_n'=:\phi_n'$ and $\chi W\chi^*=:E$, we have 
\begin{equation}\label{eq:phi'} 
\phi_{n+1}'=zE \phi_n'+\rho,\;\phi_0'=0, 
\end{equation}
where $\rho=\chi W\Phi_0$. 
The dynamical system given by (\ref{eq:phi'}) can be also accomplished by the restriction to the internal graph of the following alternating quantum walk for $z=-1$ case: 
the time evolution operator of the alternating quantum walk
$U:\mathbb{C}^{\tilde{A}} \rightarrow \mathbb{C}^{\tilde{A}}$ 
is defined in the matrix form by
\begin{equation}\label{eq:U}
(U)_{a,b}= \begin{cases}
(-1)^{\mathbbm{1}_{V_{0}}(o(a))} \left(\dfrac{2}{\text{deg}(o(a))}-\delta_{a \overline{b}} \right) & \text{ if } \, o(a)=t(b),\\
0 & \text{otherwise},
\end{cases}
\end{equation}
where $\mathbbm{1}_{V_{0}}$ is the characteristic function of $V_{0}$. 
The quantum coin assigned at every vertex in the internal graph is the ``signed" Grover matrix.  
The initial state of the walker is
\begin{equation}\label{eq:initial}
\Psi_{0}(a)= \begin{cases}
\alpha_{j} & \text{if} \, o(a)=v_{j}^{(s+1)}, \,  t(a)=v_{j}^{(s)},\, s=0,1,2,..., \, j=1,...,r,\\\\
0 & \text{otherwise}.
\end{cases}
\end{equation}
In this paper, we consider $\Psi_{n+1}=U\Psi_n$ instead of $\Phi_n$. Since 
$\Psi_\infty:=\lim_{n\to\infty}\Psi_n$ exists and $U\Psi_\infty=\Psi_\infty$ holds according to \cite{Higuchi:Segawa}, we compute  its stationary state. 
In particular, we focus on the following quantities. 
\begin{definition}\label{def:scattering}
Scattering matrix: 
Let $\bm{\alpha}:= \left[ \alpha_{1},\alpha_{2},...,\alpha_{r} \right]^{T}, \, \bm{\beta}:=\left[\beta_{1},\beta_{1},...,\beta_{r}\right]^{T}$, where $\beta_{j}=\Psi_{\infty}(a)$ with $o(a)=v_{j}^{(0)}, \, t(a)=v_{j}^{(1)}$. 
The scattering matrix $\sigma$, which is a $r$-dimensional unitary matrix, is defined by 
\[ \bm{\beta}=\sigma\bm{\alpha}. \]
\end{definition}
The existence of such a unitary matrix is ensured by \cite{Feldman:Hillery:2005,Feldman:Hillery:2007}. 
If $z=1$, the scattering matrix gives us only the information of the number of tails because the scattering matrix for $z=1$ is expressed by $\mathrm{Gr}(r)$~\cite{Higuchi:Sabri:Segawa}. In this paper, if $z=- 1$, we obtain the information on the internal graph in Theorem~1. 
We are also interested in the stationary state in the internal graph, especially how many quantum walkers exist; that is, how quantum walkers feel {\it comfortable} to the graph. 
\begin{definition}\label{def:comfortability}
The comfortability of $G_0$ with $\bm{\alpha}$ with  $\delta V$ for quantum walker : 
\[\mathcal{E}_{QW}(G_0;\bm{\alpha},\delta V)=\frac{1}{2}\sum_{a\in A_0}|\Psi_\infty(a)|^2. \]
\end{definition}
We extract some geometric graph structures from these quantities which derive from some quantum effects in Theorem~2.  


\section{Main Results}
\label{Results}
Now we state a main theorem concerning the scattering on the surface, which gives a quantum walk method to characterize bipartite graphs. For example, let us consider a graph with two tails and we set the inflow one of them. 
According to the following theorem, we can determine the bipartiteness of the graph from the scattering way: if the scattering is the perfect reflecting, then the graph is non-bipartite, while the scattering is the perfect transmitting, then the graph is bipartite. 

\begin{reptheorem}{main1}{\rm (Scattering on the surface)}
Assume the time evolution operator is described by (\ref{eq:U}). 
For the stationary state $\Psi_\infty$, let $\bm{\alpha}:= \left[ \alpha_{1},\alpha_{2},...,\alpha_{r} \right]^{T}$ and  $\bm{\beta}:=\left[\beta_{1},\beta_{1},...,\beta_{r}\right]^{T}$, where
$\alpha_j$ is the inflow described by (\ref{eq:initial}) and $\beta_j$ is the outflow described by 
$\beta_{j}=\Psi_{\infty}(a)$ with $o(a)=v_{j}^{(0)}, \, t(a)=v_{j}^{(1)}$. Then the scattering matrix, which is an $r$-dimensional unitary matrix; $\bm{\beta} = \sigma \bm{\alpha}$, is expressed as follows. 

\[ \sigma = 
\begin{cases}
I & \text{: $G_{0}$ is non-bipartite,}\\
\tau & \text{: $G_{0}$ is bipartite.}
\end{cases} \]
Here $I$ is the identity matrix and $\tau$ is described as follows: 
\[\tau=-\begin{bmatrix}I_{k} & 0 \\ 0 & -I_{r-k}\end{bmatrix} \mathrm{Gr}(r) \begin{bmatrix}I_{k} & 0 \\ 0 & -I_{r-k}\end{bmatrix}\]
with the computational basis labeled by $\left\lbrace v_1^{(0)},...,v_k^{(0)},v_{k+1}^{(0)},...,v_r^{(0)}\right\rbrace$, where $v_{1}^{(0)},...,v_{k}^{(0)} \in X \cap \delta V$ and $v_{k+1}^{(0)},...,v_{r}^{(0)} \in Y \cap \delta V$.
\end{reptheorem}
\begin{remark}
If we set the inflow to a bipartite graph with the partite sets $X$ and $Y$ so that total inflow to the partite set $X$ coincides with that to the partite set $Y$, the perfect reflecting happens. This means that we can not detect the bipartiteness of graph with such an initial state. 
These inputs vectors are described by eigenvectors of the scattering matrix $\tau$ with the eigenvalue $1$. 
\end{remark}
We express the total energy in the interior, comfortability, in terms of the combinatorial geometry of finite graphs.  
To state our claim, we prepare some notations and settings. 
The odd unicyclic graph is the graph which contains only one cycle whose length is odd. 
Now let us consider the case where two tails are connected to $G_{0}$ at $u_{1}$ and $u_{n}$ $(u_1\neq u_n)$, and set the inflow to $u_1$ and $u_n$ by $\bm{\alpha}=(\alpha_1,\alpha_2)=(1,0)$. 
Here $n$ is the number of vertices of $G_0$. 
In this case, the comfortability is described by 
$\mathcal{E}_{QW}(G_0;\bm{\alpha},\delta V)=:\mathcal{E}_{QW}(G_0;u_1,u_n)$. 

\begin{definition}\label{def:subgraphs}
{\rm (Important graph factors)}
Let $\chi_{1}(G_{0})$ be the number of spanning trees of $G_{0}$ and $\chi_{2}(G_{0};u_{1},u_{n})$ be the number of spanning forests of $G_{0}$ with exactly two components, one containing $u_{1}$ and the other containing $u_{n}$.
Here the isolate vertex is regarded as a tree. 
Define the set of odd unicyclic spanning subgraphs of $G_{0}$ by
$$
\mathcal{C}_o=\cup_{r\geq 1}\;\{\{C_{1},...,C_{r}\} \;|\; C_{i}\text{'s are odd unicyclic graphs}\}
$$
and the set of spanning subgraphs of $G_{0}$ whose one component is a tree $T$ which contains $u_{1}$ and the remaining components are odd unicyclic graphs by
$$
\mathcal{T}\mathcal{C}_o=\cup_{k\geq 0}\;\{\{T,C_{1},...,C_{k}\}\;|\; u_{1}\in T \text{ where } T \text{ is a tree and }C_{i}\text{'s are odd unicyclic graphs}\}.
$$
Here the situation $k=0$ is
the set of spanning trees of $G_{0}$.
See Figure~\ref{fig:subgraphs}. 
Now define the functions $\iota_{1}$ and $\iota_{2}$ by 
$$
\iota_{1}(G_{0})= \sum \limits _{H \in \mathcal{C}_o} 4^{\omega(H)}
$$
and 
$$
\iota_{2}(G_{0};u_{1})= \sum \limits _{H \in \mathcal{T}\mathcal{C}_o} 4^{\omega(H)-1},
$$
where $\omega(H)$ is the number of components in $H$.
\end{definition}

Then with the above notations, we have another main theorem of our paper as follows.

\begin{reptheorem}{thm:comfortability}{\rm (Comfortability in the interior)}
\label{power}
Assume the number of tails is $2$, and the inflow $\bm{\alpha}=(\alpha_1,\alpha_2)=(1,0)$ at $u_1$ and $u_n$, respectively.
Then the comfortability of the quantum walk (\ref{eq:U}) is given by
$$
\varEpsilon_{QW}(G_{0};u_1,u_n)=
\begin{cases}
\dfrac{1}{4}\left(\dfrac{\chi_{2}(G_{0};u_{1},u_{n})}{\chi_{1}(G_{0})}+ |E_{0}|\right)  & \text{ if } G_{0} \text{ is bipartite,}\\\\
\dfrac{\iota_{2}(G_{0};u_{1})}{\iota_{1}(G_{0})} & \text{ if } G_{0} \text{ is non-bipartite.}
\end{cases}
$$
\end{reptheorem}
Then we can determine how much quantum walker feels {\it comfortable} to the given graph by listing up the spanning subgraphs in Definition~\ref{def:subgraphs} of this graph. 
We will demonstrate it for the graphs with four vertices in the next section. 

\begin{theorem}\label{thm:4}
Assume the number of tails is $2$, and the inflow $\bm{\alpha}=(\alpha_1,\alpha_2)=(1,0)$ at $u_1$ and $u_n$, respectively.
Then the comfortability of the quantum walk discussed in \cite{Higuchi:Segawa}, that is, $z=1$ in (\ref{eq:phi'}), is given by
$$
\varEpsilon_{QW}(G_{0};u_1,u_n)=
\dfrac{1}{4}\left(\dfrac{\chi_{2}(G_{0};u_{1},u_{n})}{\chi_{1}(G_{0})}+ |E_{0}|\right). 
$$
\end{theorem}
The proof of Theorem~\ref{thm:4} is essentially same as in the proof of Theorem~\ref{thm:comfortability} for the bipartite graph case. 

\begin{figure}
    \centering
    \includegraphics[width=12cm]{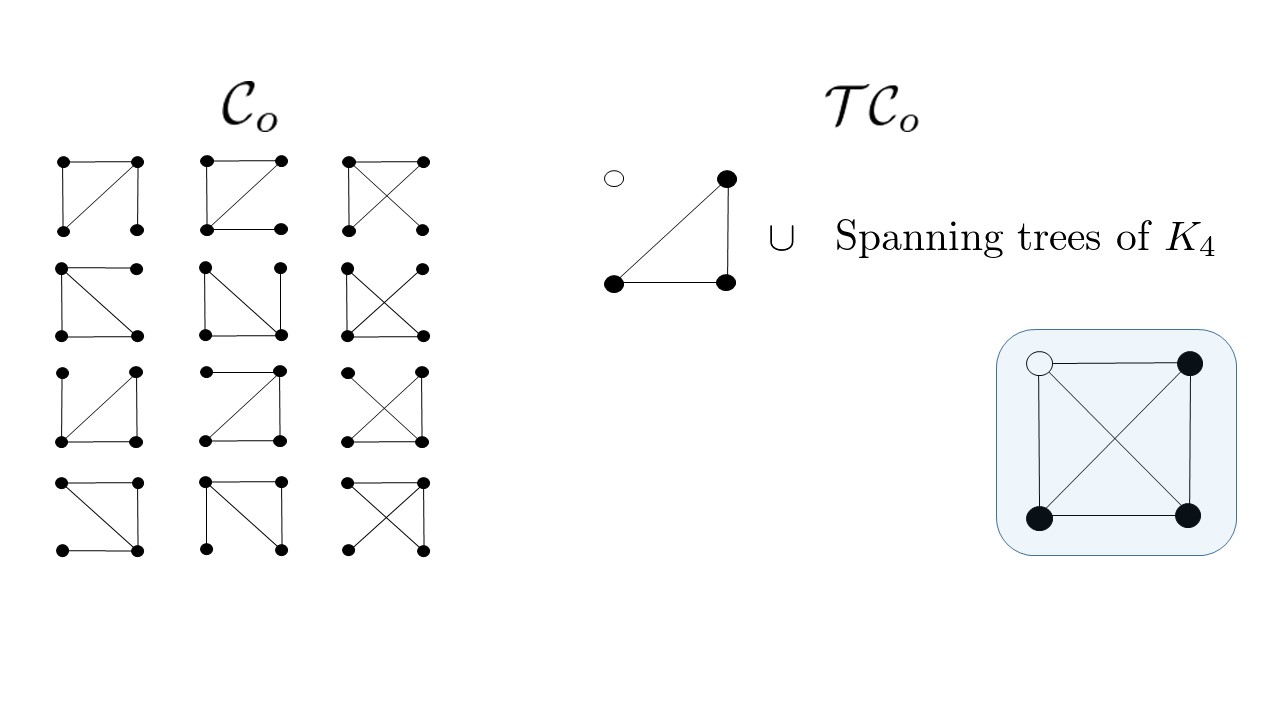}
    \caption{$\mathcal{C}_o$-factor and $\mathcal{T}\mathcal{C}_o$-factor of $K_{4}$: The white colored vertex in the complete graph $K_4$ corresponds to $u_1$. The left figure depicts the list of the family of odd unicyclic factor of $K_4$ and the right figure depicts the list of the $\mathcal{T}\mathcal{C}_o$-factor of $K_4$. Note that the isolated vertex is regards as a tree and the family of the  spanning tree is included in $\mathcal{T}\mathcal{C}_o$. }
    \label{fig:subgraphs}
\end{figure}

\section{Example}\label{Example}
As an example to our result in Theorem \ref{power}, we consider the connected  graphs with $4$ vertices labeled by  $\left\lbrace u_{1},u_{2},u_{3},u_{4}\right\rbrace$. The setting is the same as in Theorem \ref{power} and we choose the inflow $\bm{\alpha}=(1,0)$ at the vertices $u_{1}$ and $u_{4}$, respectively. We classify these graphs into $10$ classes based on the number of edges, bipartiteness and the configuration of $u_{1}$ and $u_{4}$ and the numbers of important factors state in Definition~\ref{def:subgraphs}. 
Note that the comfortability for the non-bipartite case depends {\it only} on $u_1$. 
We conclude that every graph with $4$ vertices belongs to exactly one of the following classes (see Fig \ref{fig:ordering}):
$$\mathcal{G}_{1}=\left\lbrace K_{4} \right\rbrace,$$
$$\mathcal{G}_{2}=\left\lbrace K_{4}- u_{1} u_{j}  : j=2,3,4 \right\rbrace,$$
$$\mathcal{G}_{3}=\left\lbrace K_{4} - u_{i} u_{j}  : i,j=2,3,4 \right\rbrace,$$
$$\mathcal{G}_{4}=\left\lbrace C_{4}:u_{1} \sim u_{4} \right\rbrace,$$
$$\mathcal{G}_{5}=\left\lbrace C_{4}:u_{1} \nsim u_{4} \right\rbrace,$$
$$\mathcal{G}_{6}=\left\lbrace G: G \text{ is  constructed by joining } u_{1} \text{ to exactly one vertex in the cycle } u_{2},u_{3},u_{4}\right\rbrace,$$
$$\mathcal{G}_{7}=\left\lbrace G:  G \text{ is constructed by joining } u_{i} (i=2,3,4) \text{ to exactly one vertex in the cycle } u_{1},u_{k},u_{l} (k,l \neq 1,i)\right\rbrace,$$
$$\mathcal{G}_{8}=\left\lbrace T:T \text{ is a tree with } \text{dist}(u_{1},u_{4})=1 \right\rbrace,$$
$$\mathcal{G}_{9}=\left\lbrace T:T \text{ is a tree with } \text{dist}(u_{1},u_{4})=2 \right\rbrace,$$
$$\mathcal{G}_{10}=\left\lbrace T:T \text{ is a tree with } \text{dist}(u_{1},u_{4})=3 \right\rbrace.$$
Here $K_{4}- u_{i} u_{j} $ is the graph obtained by removing the edge $ u_{i} u_{j} $ from $K_{4}$. Now for $G_{i} \in \mathcal{G}_{i}(i=1,...,10)$, we compute $\varEpsilon_{QW}$ as follows. When $G_{0}=K_{4}$ for example, since $K_{4}$ is non-bipartite, we have 
\[
\varEpsilon_{QW}(K_{4})=\dfrac{\iota_{2}(K_{4};u_{1})}{\iota_{1}(K_{4})}.
\]
To compute $\iota_{1}(K_{4})$, we need to find the number of odd unicyclic subgraphs which span $K_{4}$. The only such possible  subgraph is a $3$-cycle with an additional edge. It is clear that there are $4$ ways to choose a $3$-cycle and for a chosen $3$-cycle, there are $3$ ways to choose an edge which connects the remaining vertex to the cycle. Hence, altogether there are $12$ such subgraphs, which are shown as $\mathcal{C}_{o}$ in the figure \ref{fig:subgraphs}. Observe that each subgraph has only one component and hence we have
\[
\iota_{1}(K_{4})=48.
\]

Now to compute $\iota_{2}(K_{4};u_{1})$, we have to find the spanning subgraphs which contains a tree with $u_{1}$ and the remaining are odd cycles, which are possibly empty. There are two types of such subgraphs, as shown as $\mathcal{T}\mathcal{C}_{o}$ in the figure \ref{fig:subgraphs} one being the spanning trees and the other being the cycle $u_{2},u_{3},u_{4}$ along with the single vertex $u_{1}$. In the first case the number of spanning trees is the tree number, or the complexity, of $K_{4}$, given by $4^{4-2}=4^{2}$ (e.g.,~\cite{Biggs}) while in the second case there is only one such subgraph which has two components. So it follows that
\[
\iota_{2}(K_{4};u_{1})=20
\]
and hence 
\[
\varEpsilon_{QW}(K_{4})=\dfrac{20}{48}=\dfrac{5}{12}.
\]

Now consider $G_{4} \in \mathcal{G}_{4}$. This graph is a $4$-cycle with $u_{1} \sim u_{4}$. This graph is bipartite and hence \
\[
\varEpsilon_{QW}(G_{4})=\dfrac{1}{4}\left(
\dfrac{\chi_{2}(G_{4};u_{1},u_{4})}{\chi_{1}(G_{4})}+|E_{0}|\right).
\]
Since $\chi_{1}(G_{4})$ is the tree number of $G_{4}$ and by removing an edge from $G_{4}$ we can get a spanning tree, we have
\[
\chi_{1}(G_{4})=4.
\]
To compute $\chi_{2}(G_{4};u_{1},u_{4})$, we have to find the number of forests with exactly two components, one containing $u_{1}$ and the other containing $u_{4}$. To find such a forest, the edge $\left\lbrace u_{1},u_{4} \right\rbrace$ has to be removed, and other than that, one of the remaining edges has to be removed. So we have
\[
\chi_{2}(G_{4};u_{1},u_{4})=3
\]
and hence
\[
\varEpsilon_{QW}(G_{4})= \dfrac{1}{4}(\dfrac{3}{4}+4)=\dfrac{19}{16}.
\]
Similarly, we compute the comfortability on the remaining classes of graphs and we tabulate these values as shown in Table 1.
\begin{table}[htbp]
\begin{center}
\begin{tabular}{| c || c |c| c| c| c| c| c| c| c| c|}
\hline
 $\mathcal{G}_{i}$ & $\mathcal{G}_{1}$ & $\mathcal{G}_{2}$ & $\mathcal{G}_{3}$ & $\mathcal{G}_{4}$ & $\mathcal{G}_{5}$ & $\mathcal{G}_{6}$ & $\mathcal{G}_{7}$ & $\mathcal{G}_{8}$ & $\mathcal{G}_{9}$ & $\mathcal{G}_{10}$ \\ 
 \hline\hline
 $\varEpsilon_{QW}(\mathcal{G}_{i})$& $5/12$ & $3/4$ & $1/2$ &  $19/16$ & $5/4$ & $7/4$ & $3/4$ & $1$ & $5/4$ & $3/2$  
\\
\hline
Scattering & R & R & R & T & T & R & R & T & T & T \\ \hline
Bipartiteness & - & - & - & $\circ$ & $\circ$ & - & - & $\circ$ & $\circ$ & $\circ$ \\ \hline
$|E|$ & $6$ & $5$ & $5$ & $4$ & $4$ & $4$ & $4$ & $3$ & $3$ & $3$ \\ \hline
\end{tabular}
\end{center}
\caption{The comfortability of quantum walker to graphs with four vertices: The comfortability to each graph class is described by $\mathcal{E}_{QW}$. The symbols of ``R" and ``T" mean the perfectly reflection and transmitting, respectively. The best graph with four vertices of the comfortability is the $\mathcal{G}_6$ and the worst graph is the complete graph.}
\end{table}

Based on $\varEpsilon_{QW}$, we have the following ordering of graphs:
$$G_{6} \succ_{QW} G_{10}\succ_{QW} G_{5}, G_{9} \succ_{QW} G_{4} \succ_{QW} G_{8}\succ_{QW} G_{2}, G_{7}\succ_{QW} G_{3}\succ_{QW} G_{1}.$$
\begin{figure}[htbp]
    \centering
    \includegraphics[width=15cm]{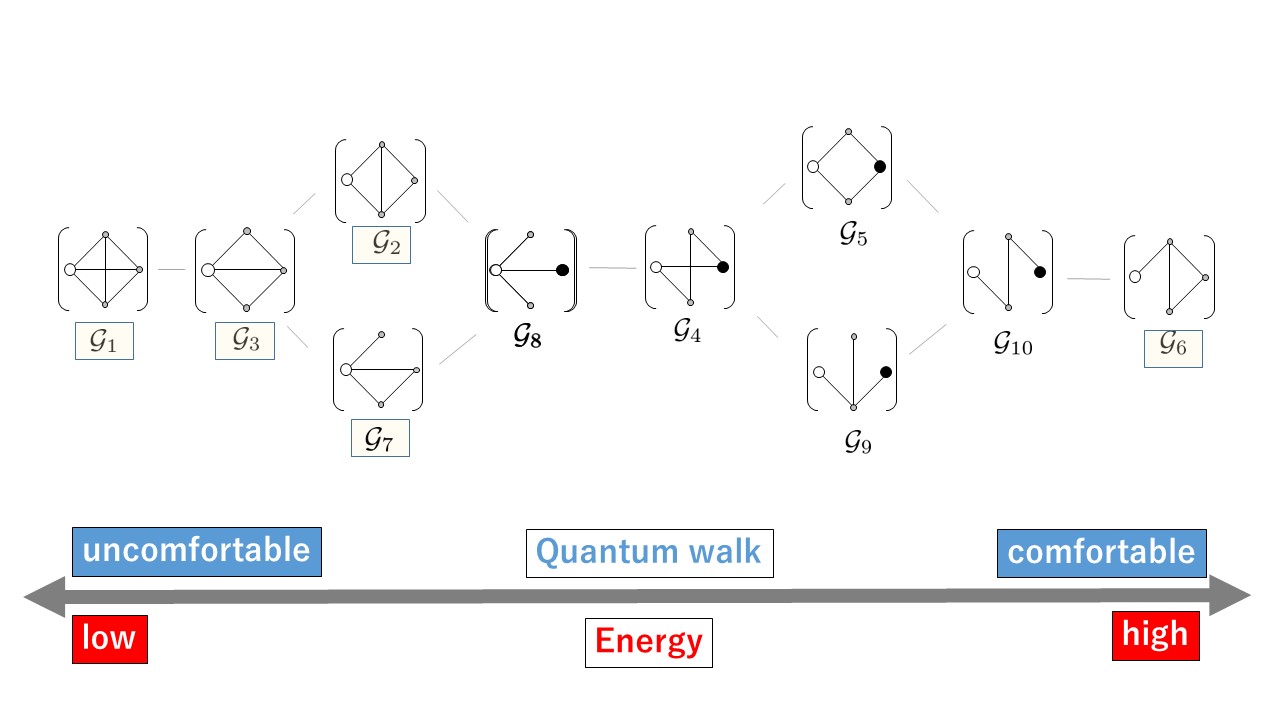}
    \caption{The Hasse diagram of comfortability of graphs:  The most comfortable graph for this quantum walker is $\mathcal{G}_{6}$ which is a non-bipartite graph, while the most uncomfortable graph is $\mathcal{G}_1$ (the complete graph). Here the entrance vertex $u_{1}$ is indicated by the white vertex, and the exit vertex $u_N$ is indicated by the black vertex for the bipartite case. Note that in the non-bipartite case, the comfortability for this quantum walker is independent of the position of the exit vertex (see Theorem~\ref{thm:comfortability}). }
    \label{fig:ordering}
\end{figure}
Furthermore, we remark that for a tree $T$ in general, with $n$ vertices, the comfortability of the quantum walker is given by 
$$\varEpsilon_{QW}(T)=\dfrac{1}{4}\left(\text{dist}(u_{1},u_{n})+  (n-1)\right).$$

Finally we also consider the comparison between the comfortability for $z=1$ and $z=-1$ cases.
Any graph of four vertices is isomorphic to one of the following graphs $\Gamma_1$,$\Gamma_2$,\dots, $\Gamma_6$ as in Figure \ref{Fig:graphs4}:

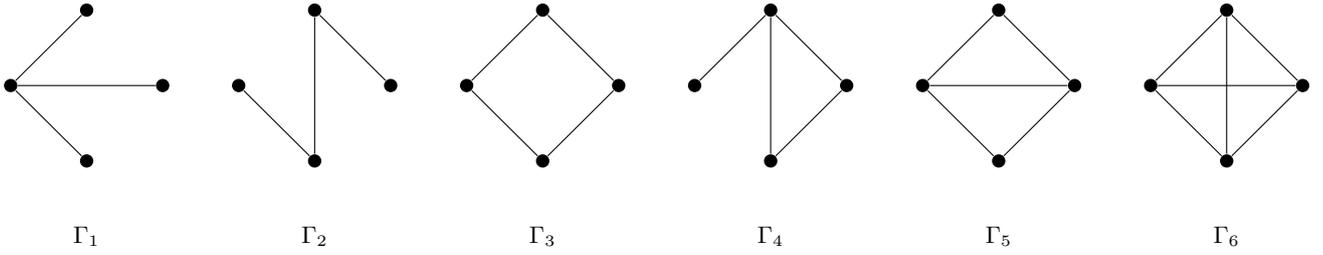
\begin{figure}[h]\label{Fig:graphs4}
\begin{center}
\begin{tikzpicture}[shorten >=1pt,->]
  \tikzstyle{vertex}=[circle,fill=black,minimum size=5pt,inner sep=1pt]
  \tikzstyle{ann} = [fill=white,font=\footnotesize,inner sep=1pt]

\node[vertex][] (G_1_1) at (-10,0) {};
  \node[vertex][] (G_1_2) at (-9,1) {};
  \node[vertex][] (G_1_3) at (-8,0) {};
  \node[vertex][] (G_1_4) at (-9,-1) {};
  \draw (G_1_1)--(G_1_2)--cycle;
  \draw (G_1_1)--(G_1_3)--cycle;
  \draw (G_1_1)--(G_1_4)--cycle;
  \node[ann] at (-9,-2){$\Gamma_1$};
  
  \node[vertex][] (G_2_1) at (-7,0) {};
  \node[vertex][] (G_2_2) at (-6,1) {};
  \node[vertex][] (G_2_3) at (-5,0) {};
  \node[vertex][] (G_2_4) at (-6,-1) {};
  \draw (G_2_1)--(G_2_4)--(G_2_2)--(G_2_3)--cycle;
  \node[ann] at (-6,-2){$\Gamma_2$};
  
  \node[vertex][] (G_3_1) at (-4,0) {};
  \node[vertex][] (G_3_2) at (-3,1) {};
  \node[vertex][] (G_3_3) at (-2,0) {};
  \node[vertex][] (G_3_4) at (-3,-1) {};
  \draw (G_3_1)--(G_3_2)--(G_3_3)--(G_3_4)--(G_3_1)--cycle;
  \node[ann] at (-3,-2){$\Gamma_3$};
  
   \node[vertex][] (G_4_1) at (-1,0) {};
  \node[vertex][] (G_4_2) at (0,1) {};
  \node[vertex][] (G_4_3) at (1,0) {};
  \node[vertex][] (G_4_4) at (0,-1) {};
  \draw (G_4_1)--(G_4_2)--(G_4_3)--(G_4_4)--(G_4_2)--cycle;
  \node[ann] at (0,-2){$\Gamma_4$};
  
   \node[vertex][] (G_5_1) at (2,0) {};
  \node[vertex][] (G_5_2) at (3,1) {};
  \node[vertex][] (G_5_3) at (4,0) {};
  \node[vertex][] (G_5_4) at (3,-1) {};
  \draw (G_5_1)--(G_5_2)--(G_5_3)--(G_5_4)--(G_5_1)--cycle;
  \draw (G_5_1)--(G_5_3)--cycle;
  \node[ann] at (3,-2){$\Gamma_5$};
  
  \node[vertex][] (G_6_1) at (5,0) {};
  \node[vertex][] (G_6_2) at (6,1) {};
  \node[vertex][] (G_6_3) at (7,0) {};
  \node[vertex][] (G_6_4) at (6,-1) {};
  \draw (G_6_1)--(G_6_2)--(G_6_3)--(G_6_4)--(G_6_1)--cycle;
  \draw (G_6_1)--(G_6_3)--cycle;
  \draw (G_6_2)--(G_6_4)--cycle;
  \node[ann] at (6,-2){$\Gamma_6$};
  
\end{tikzpicture}
\end{center}
\caption{Non-isomorphic graphs $\Gamma_j$'s for four vertices}
\end{figure}

We set 
\[ \mathrm{Comf}(G)=\max_{u_1,u_n\in V(G)} \mathcal{E}_{QW}(G;\bm{\alpha},u_1,u_n) \]
for $\bm{\alpha}=(1,0)$ at $u_1$ and $u_n$. 
Then we have the comfortabilities for each graph and $z=\pm 1$ as in Table~2. 
\begin{table}[htbp]\label{table:comf2}
\begin{center}
\begin{tabular}{| c || c |c| c| c| c| c|}
\hline
  & $\Gamma_{1}$ & $\Gamma_{2}$ & $\Gamma_{3}$ & $\Gamma_{4}$ & $\Gamma_{5}$ & $\Gamma_{6}$  \\ 
 \hline\hline
 $z=-1$ & $5/4$ & $3/2$ & $5/4$ &  $7/4$ & $3/4$ & $5/12$   
\\
\hline
 $z=1$ & $5/4$ & $3/2$ & $5/4$ &  $17/12$ & $3/2$ & $13/8$   
\\ \hline
\end{tabular}
\end{center}
\caption{$\mathrm{Comf}(\Gamma_j)$ for $z=\pm 1$}
\end{table}
Thus we have, for $z=-1$, 
    \[\mathrm{Comf}(\Gamma_6)<\mathrm{Comf}(\Gamma_5)<\mathrm{Comf}(\Gamma_1)=\mathrm{Comf}(\Gamma_3)<\mathrm{Comf}(\Gamma_2)<\mathrm{Comf}(\Gamma_4), \]
and for $z=1$,
\[\mathrm{Comf}(\Gamma_1)=\mathrm{Comf}(\Gamma_3)<\mathrm{Comf}(\Gamma_4)<\mathrm{Comf}(\Gamma_2)=\mathrm{Comf}(\Gamma_5)<\mathrm{Comf}(\Gamma_6). \]

\section{Proof of main theorems}
\label{Proof}
\subsection{Proof of Theorem \ref{main1}}

To prove this theorem we use the following lemma.

\begin{lemma}
\label{lemma}
For a given $u \in V_{0}$, $\Psi_{\infty}(a)-\Psi_{\infty}(\overline{a})$ is constant for all $a \in \tilde{A}$ with $o(a)=u$.
\end{lemma}
\begin{proof}
It follows from the dynamics of the walk that, in the stationary state, for $a \in \tilde{A}$ such that $o(a) \in V_{0}$,

$$\Psi_{\infty}(a)=-\sum \limits_{b:t(b)=o(a)} \dfrac{2}{\text{deg}(o(a))}\Psi_{\infty}(b)+\Psi_{\infty}(\overline{a}).$$
It follows that
$$\Psi_{\infty}(a)-\Psi_{\infty}(\overline{a})=-\dfrac{2}{\text{deg}(u)}\sum \limits_{b:t(b)=u} \Psi_{\infty}(b), \; o(a)=u.$$
Observe that for a given $u \in V_{0}$, the right hand side of the equation is a constant. 
\end{proof}

Now we give the proof of Theorem \ref{main1} as follows. By Lemma \ref{lemma}, the measure $\Psi_{\infty}(a)-\Psi_{\infty}(\overline{a})$ is a measure on $u \in V_{0}$ such that $o(a)=u$. We denote this measure by $\rho(u)$. That is, $\Psi_{\infty}(a)-\Psi_{\infty}(\overline{a})=\rho(u), o(a)=u$. It follows that if $u \sim v$ then $\rho(u)=-\rho(v)$.

Suppose $G_{0}$ is non-bipartite. Then there is an odd cycle $C=(u_{1},...,u_{2l-1}), u_{i} \in V_{0}$ in $G_{0}$. Then we have $\rho(u_{1})=-\rho(u_{2})=\rho(u_{3})=...=\rho(u_{2l-1})=-\rho(u_{1})$ which implies $\rho(u_{1})=0$. Since $G_{0}$ is connected, $\rho(u)=0$ for any $u \in V_{0}$. Thus $\Psi_{\infty}(a)-\Psi_{\infty}(\overline{a})=0$ for any $a \in A_{0}$ such that $o(a) \in V_{0}$. In particular $\beta_{i}=\alpha_{i}$ for any $i$. Hence $\bm{\beta}=\bm{\alpha}$.

Now suppose $G_{0}$ is bipartite with $V_{0}=X\sqcup Y$. Observe that
\begin{equation}
\label{X}
s := \rho_{V_{0}}(v)
\end{equation}
is constant for all $v \in X$. Then
\begin{equation}
\label{Y}
\rho_{V_{0}}(v)=-s
\end{equation}
for all $v \in Y$. Define
\begin{align*}
A^{X}_{in} &= \left\lbrace a \in \tilde{A} : o(a) \in \tilde{V}\setminus V_{0} , t(a) \in X \right\rbrace,\\
A^{X}_{out} &= \left\lbrace a \in \tilde{A} : \overline{a} \in A^{X}_{in} \right\rbrace,\\
A^{Y}_{in} &= \left\lbrace a \in \tilde{A} : o(a) \in \tilde{V}\setminus V_{0} , t(a) \in Y \right\rbrace,\\
A^{Y}_{out} &= \left\lbrace a \in \tilde{A} : \overline{a} \in A^{Y}_{in} \right\rbrace.
\end{align*}
Then by applying the time evolution operator once, we get
$$U \left( \Psi_{\infty}|_{A^{X}_{in}}+\Psi_{\infty}|_{A^{Y}_{in}}+\Psi_{\infty}|_{A_{0}}\right) = \Psi_{\infty}|_{A^{X}_{out}}+\Psi_{\infty}|_{A^{Y}_{out}}+\Psi_{\infty}|_{A_{0}}.$$
By taking the squared norm, we get
$$\left\lVert \Psi_{\infty}|_{A^{X}_{in}} \right\rVert^{2}+\left\lVert \Psi_{\infty}|_{A^{Y}_{in}} \right\rVert^{2}+\left\lVert \Psi_{\infty}|_{A_{0}} \right\rVert^{2}=\left\lVert \Psi_{\infty}|_{A^{X}_{out}} \right\rVert^{2}+\left\lVert \Psi_{\infty}|_{A^{Y}_{out}} \right\rVert^{2}+\left\lVert \Psi_{\infty}|_{A_{0}} \right\rVert^{2}.$$
It follows from (\ref{X}) and (\ref{Y}) that
\begin{align*}
\left\lVert \Psi_{\infty}|_{A^{X}_{in}} \right\rVert^{2}+\left\lVert \Psi_{\infty}|_{A^{Y}_{in}} \right\rVert^{2}&= \left\lVert s \mathbbm{1}_{A^{X}_{in}}+S \Psi_{\infty}|_{A^{X}_{in}} \right\rVert^{2}+\left\lVert -s \mathbbm{1}_{A^{Y}_{in}}+S \Psi_{\infty}|_{A^{Y}_{in}} \right\rVert^{2}\\
&=s^{2} |X \cap \delta V| + 2s \sum \limits_{a \in A^{X}_{in}} \Psi_{\infty}(a) +\left\lVert \Psi_{\infty}|_{A^{X}_{in}} \right\rVert^{2}\\
&\qquad\qquad +s^{2} |Y \cap \delta V| 
- 2s \sum \limits_{a \in A^{Y}_{in}} \Psi_{\infty}(a)+\left\lVert \Psi_{\infty}|_{A^{Y}_{in}} \right\rVert^{2}.
\end{align*}
Here $S$ is the shift operator: $S\delta_a=\delta_{\bar{a}}$ for any $a\in \tilde{A}$. 
Hence if $s\neq 0$,
\begin{equation}
\label{constant}
    s=-\dfrac{2}{|\delta V|}\left( \sum \limits_{a \in A^{X}_{in}} \Psi_{\infty}(a)- \sum \limits_{a \in A^{Y}_{in}} \Psi_{\infty}(a)\right).
\end{equation}
Then it follows that
$$\beta_{i} = \alpha_{i} -\dfrac{2}{|\delta V|}\left( \sum \limits_{a \in A^{X}_{in}} \Psi_{\infty}(a)- \sum \limits_{a \in A^{Y}_{in}} \Psi_{\infty}(a)\right),$$
where $\beta_{i}=\Psi_{\infty}(a)$ for some $a \in  A^{X}_{out}$ and
$$\beta_{i} = \alpha_{i} -\dfrac{2}{|\delta V|}\left( \sum \limits_{a \in A^{Y}_{in}} \Psi_{\infty}(a)- \sum \limits_{a \in A^{X}_{in}} \Psi_{\infty}(a)\right),$$
where $\beta_{i}=\Psi_{\infty}(a)$ for some $a \in  A^{Y}_{out}$. Hence $\bm{\beta}=\tau \bm{\alpha}$ where
$$
\tau = 
\left(\begin{array}{@{}c|c@{}}
  \begin{matrix}
  -\dfrac{2}{|\delta V|}+1 & -\dfrac{2}{|\delta V|}&... \\
  -\dfrac{2}{|\delta V|} & -\dfrac{2}{|\delta V|}+1 & ...\\
  . & & \\
  . & & \\
  . & & \\
  \end{matrix}
  & \bignumber \\
\hline
  \bignumber &
  \begin{matrix}
  -\dfrac{2}{|\delta V|}+1 & -\dfrac{2}{|\delta V|}&... \\
  -\dfrac{2}{|\delta V|} & -\dfrac{2}{|\delta V|}+1 & ...\\
  . & & \\
  . & & \\
  . & & \\
  \end{matrix}
\end{array}\right)
.$$

It is clear that, to satisfy the condition $s\neq 0$, we must have $ \sum_{a \in A^{Y}_{in}} \Psi_{\infty}(a)\neq \sum _{a \in A^{X}_{in}} \Psi_{\infty}(a)$. 
Now let us see that $s=0$ if and only if $\kappa_X=\kappa_Y$, where 
$\kappa_X=\sum_{a \in A^{X}_{in}} \Psi_{\infty}(a)$ and $\kappa_Y=\sum_{a \in A^{Y}_{in}} \Psi_{\infty}(a)$. 


By Lemma \ref{lemma} we have, $s=0$ if and only if
$$\Psi_{\infty}(a)=\Psi_{\infty}(\overline{a}),\;a \in A_0 \text{ with } o(a) \in V_{0}$$
and $$\sum \limits_{b:t(b)=u} \Psi_{\infty}(b) =0,\; u \in V_{0}. $$
Thus we have 
\begin{align*}
0 &=\sum \limits_{a:t(a) \in Y} \Psi_{\infty}(a) \\ 
&=\sum_{a:\; t(a)\in Y,\;o(a)\in X}\Psi_{\infty}(a)+\kappa_Y \\
&=\sum_{a:\in o(a)\in X,\;t(a)\in Y} \Psi_{\infty}(a)+\kappa_Y \\
&=\sum \limits_{a:o(a) \in X} \Psi_{\infty}(a)-\kappa_X+\kappa_Y \\
&= -\kappa_X+\kappa_Y,
\end{align*}
which implies $\kappa_X=\kappa_Y$ if $s=0$. 

If $\bm{\alpha}$ satisfies $\kappa_X=\kappa_Y$, then $s=0$ and the same argument as in the case where $G_{0}$ is non-bipartite is valid;  $\bm{\beta}=\bm{\alpha}$ holds. It is easy to check that $\bm{\beta}=\tau \bm{\alpha}$ holds for $\tau$ defined above.
$\hfill \square$

\begin{figure}
\begin{center}
\begin{tikzpicture}[shorten >=1pt,->]
  \tikzstyle{vertex}=[circle,fill=black,minimum size=5pt,inner sep=1pt]
  \tikzstyle{vertex}=[circle,fill=black,minimum size=5pt,inner sep=1pt]
  \tikzstyle{defective}=[circle,fill=blue,minimum size=5pt,inner sep=1pt]
  \tikzstyle{ann} = [fill=white,font=\footnotesize,inner sep=1pt]
  
\node[vertex][] (G_1) at (-2,2) {};
  \node[vertex][] (G_2) at (-2,1) {};
  \node[vertex][] (G_3) at (-2,0) {};
  \node[vertex][] (G_4) at (-2,-1) {};

  \node[vertex][] (G_6) at (2,2) {};
  \node[vertex][] (G_7) at (2,1) {};
  \node[vertex][] (G_8) at (2,0) {};
  \node[vertex][] (G_9) at (2,-1) {};
 
  \node[vertex][] (G_1_1) at (-3.5,2) {};
  \node[vertex][] (G_2_1) at (-3.5,1) {};
  \node[vertex][] (G_3_1) at (-3.5,0) {};
  \node[vertex][] (G_4_1) at (-3.5,-1) {};
 
  \node[vertex][] (G_6_1) at (3.5,2) {};
  \node[vertex][] (G_7_1) at (3.5,1) {};
  \node[vertex][] (G_8_1) at (3.5,0) {};
  \node[vertex][] (G_9_1) at (3.5,-1) {};
  
  \draw[draw=black] (-2.5,-1.5) rectangle ++(1,4);
  \draw[draw=black] (1.5,-1.5) rectangle ++(1,4);
  
  \draw (G_1_1)--(G_1)--cycle;
  \draw (G_2_1)--(G_2)--cycle;
  \draw (G_3_1)--(G_3)--cycle;
  \draw (G_4_1)--(G_4)--cycle;
  \draw (G_6_1)--(G_6)--cycle;
  \draw (G_7_1)--(G_7)--cycle;
  \draw (G_8_1)--(G_8)--cycle;
  \draw (G_9_1)--(G_9)--cycle;
  \draw (G_1)--(G_6)--cycle;
  \draw (G_2)--(G_6)--cycle;
  \draw (G_3)--(G_8)--cycle;
  \draw (G_4)--(G_7)--cycle;
  \draw (G_1)--(G_9)--cycle;
  \draw (G_2)--(G_8)--cycle;
  \draw (G_3)--(G_7)--cycle;
  \draw (G_4)--(G_6)--cycle;
  
  \draw[arrows=->,color=red,ultra thick](-3,0)--(-2,0);

  \node[ann] at (-2.75,0.25){$1$};

\end{tikzpicture}
\caption{Example of the inflow of $s\neq 0$. }
\end{center}
\end{figure}
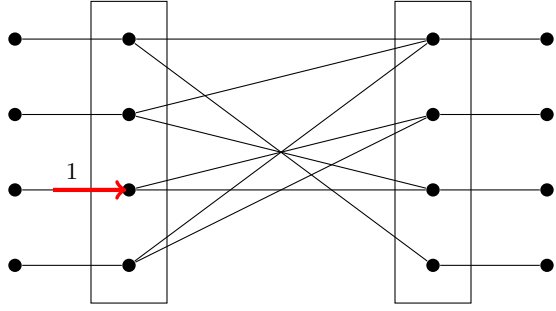

\subsection{Proof of Theorem \ref{thm:comfortability}}
The proof of Theorem~\ref{thm:comfortability} is divided into three parts with bipartite and non-bipartite cases as follows. 
First, we will show that the stationary state is expressed by a (pseudo-)current function. 
Secondly, we will determine the potential functions with respect to the (pseudo-)current function which satisfies the Poisson equation in terms of the (signless-)Laplacian. 
Finally, we will express the comfortability in terms of geometries of the graph using the above. 

\subsubsection{Kirchhoff and pseudo-Kirchhoff laws}
In this subsection we will see that the stationary state can represent a kind of current function \cite{Higuchi:Sabri:Segawa, Higuchi:Segawa} on the underlying bipartite and non-bipartite graphs. \\

\noindent{\bf Bipartite case} \\
First, we will introduce a current function induced by the stationary state $\Psi_\infty$ on a bipartite graph. Let $G_{0}$ be a bipartite graph with $V_{0}=X \sqcup Y$. Define the function $f(\cdot)$ such that
$$f(a)= 
\begin{cases}
0 & \text{if } \, t(a) \in X,\\\\
1 & \text{if } \, t(a) \in Y.
\end{cases}
$$
Then by Lemma \ref{lemma}, it follows that the measure $\dfrac{1}{2} \left( (-1)^{f(a)} \Psi_{\infty}(a) + (-1)^{f(\overline{a})} \Psi_{\infty}(\overline{a}) \right)$ is a constant for any $a \in A_{0}$. 
We denote this constant by $\rho$. Then the following theorem holds.
\begin{theorem}
\label{main2}
Let the setting be the same as the above and define $j(a)=(-1)^{f(a)} \Psi_{\infty}(a)-\rho$. Then $j(\cdot)$ satisfies the Kirchhoff's current and voltage laws: 
\begin{enumerate}
\item (Kirchhoff current law) 
For any $u \in \tilde{V}$ and $a \in \tilde{A}$, 
$$\sum \limits_{b \in \tilde{A}:t(b)=u} j(b)=0, \, j(a)+j(\overline{a})=0.$$
\item (Kirchhoff voltage law)
For any cycle $c=(a_{1},...,a_{s})$ with $t(a_{1})=o(a_{2}),...,t(a_{s-1})=o(a_{s}),t(a_{s})=o(a_{1})$ in $G_{0}$, it holds
$$\sum \limits_{k=1}^{s} j(a_{k})=0.$$
\end{enumerate}

\end{theorem}

\begin{proof}
We can rewrite the expression in Lemma \ref{lemma} as
$$\dfrac{1}{2} \left( (-1)^{f(a)} \Psi_{\infty}(a) + (-1)^{f(\overline{a})} \Psi_{\infty}(\overline{a}) \right)=\dfrac{1}{\text{deg}(u)}\sum \limits_{b:t(b)=u} (-1)^{f(b)} \Psi_{\infty}(b), \, o(a)=u,$$
in which the right hand side is a measure on the vertex $u$, denoted by $\rho(u)$. Then it follows that if $u \sim v$ then $\rho(u)=\rho(v)$. Since $G_{0}$ is connected, it follows that $\rho(u)$ is a constant for all $u \in V_{0}$, which is denoted by $\rho$. Summarizing the above, we have
\begin{equation}\label{eq:rho}
\dfrac{1}{\text{deg}(u)}\sum \limits_{b:t(b)=u} (-1)^{f(b)} \Psi_{\infty}(b)=\rho, u \in V_{0}.
\end{equation}
Then from (\ref{constant}) we can express
\[
\rho=\dfrac{1}{|\delta V|}\left( \sum \limits_{a \in A^{X}_{in}} \Psi_{\infty}(a)- \sum \limits_{a \in A^{Y}_{in}} \Psi_{\infty}(a)\right).
\]
Equation (\ref{eq:rho}) implies 
$$\sum \limits_{a:t(a)=u}j(a)=\sum \limits_{a:t(a)=u} (-1)^{f(a)} \Psi_{\infty}(a)-\rho \text{deg}(u)=0,\;u \in V_{0}.$$
Also we have
$$j(a)+j(\overline{a})=(-1)^{f(a)} \Psi_{\infty}(a) + (-1)^{f(\overline{a})} \Psi_{\infty}(\overline{a}) - 2\rho=0, \;a \in A_{0}.$$
Now suppose $C=(a_{1},...,a_{2l}),a_{i} \in A_{0}$ is an even cycle in $G_{0}$. Assume that $o(a_{1}) \in X$. Define $\varphi$ such that
$$\varphi(a)=
\begin{cases}
1 & \text{if } a = a_{2k+1} \text{ or } \overline{a}_{2k+1},\\
-1 & \text{if } a = a_{2k} \text{ or } \overline{a}_{2k},\\
0 & \text{otherwise.}
\end{cases}
$$
Let $\psi_\infty:=\chi^*\Psi_\infty$.
Since $\varphi$ is a centered eigenvector of $\chi U \chi^{*}$ whose eigenvalue is $-1$, by \cite[Lemmas 3.4 and 3.5]{Higuchi:Segawa} we have $\psi_{\infty} \perp \varphi$ and it follows that 
\begin{align*}
0&=\left\langle \psi_{\infty} | \varphi \right\rangle\\
&=\sum \limits_{a \in A_{0}} \psi_{\infty}(a)\varphi(a)\\
&=\sum \limits_{k=1}^{l} \left[\left[ \psi_{\infty}(a_{2k-1})+\psi_{\infty}(\overline{a}_{2k-1})\right]-\left[ \psi_{\infty}(a_{2k-1})+\psi_{\infty}(\overline{a}_{2k-1})\right]\right]\\
&=\sum \limits_{k=1}^{l} \left[(-1)^{f(a_{2k-1})} \left( j(a_{2k-1})-j(\overline{a}_{2k-1})\right)-(-1)^{f(a_{2k})} \left( j(a_{2k})-j(\overline{a}_{2k})\right) \right]\\
&=\sum \limits_{k=1}^{2l} \left[j(a_{k})-j(\overline{a}_{k}) \right].
\end{align*}
Since $j(a)+j(\overline{a})=0$, it follows that
$$\sum \limits_{k=1}^{s} j(a_{k})=0.$$
\end{proof}

\quad\\
\noindent{{\bf Non-bipartite case}}\\
Next let us investigate the property of $\psi_\infty$ in the case non-bipartite case. We can see that the stationary state itself has similar properties to the electrical current flow as follows. 
\begin{theorem}\label{thm:pKL}
Let $\Psi_\infty\in \tilde{A}$ be the stationary state on the non-bipartite graph. Then $\Psi_\infty$ satisfies the following properties: 
\begin{enumerate}
    \item (Pseudo-Kirchhoff current law) For each $u\in V_0$ and $a\in \tilde{A_0}$, 
    \[ \sum_{b:t(a)=u}\Psi_\infty(a)=0,\;\Psi_\infty(a)=\Psi_\infty(\bar{a}). \]
    \item (Pseudo-Kirchhoff voltage law) 
    For any even closed walk $(e_1,\dots,e_{2s})$ such that \\ $t(e_1)=o(e_2),\dots,t(e_{2s-1})=o(e_{2s})$ and $t(e_{2s})=o(e_1)$, 
    \[ \sum_{k=1}^{2s}(-1)^k\Psi_\infty(e_{k})=0. \]
\end{enumerate}
\end{theorem}
\begin{proof}
The proof of the first part is obtained  from the proof of Theorem~\ref{main1}.
The second part shows $\langle  \psi_\infty,\phi \rangle=0$ for any $\phi$ which is an eigenfunction of the eigenvalue $-1$ in the birth part. For details, refer to \cite[Definition 4 and Theorem 1]{HKSS}. Moreover it has been shown by \cite{Higuchi:Segawa} that the stationary state is orthogonal to the birth part. Then we have the desired conclusion. 
\end{proof}

\subsubsection{Laplacian and signless Laplacian}
We have seen that the stationary state has the properties of current function or pseudo-current function in $\mathbb{C} ^ {\tilde{A}}$. 
Then it is natural to determine the potential function in $\mathbb{C}^{V_0}$ with respect to the current.  
In this subsection, we characterize the  
potential function using the Laplacian and the signless Laplacian for the cases of bipartite and non-bipartite graphs, respectively. \\

\noindent{\bf Bipartite case}\\
Let $M$ be the adjacency matrix and $D$ be the degree matrix of $G_0$. The Laplacian matrix  of $G_0$ is denoted by $L=M-D$. 
Using the Laplacian matrix with the Poisson equation, 
we can characterise the current function $j(\cdot)$ on $A_0$ in terms of a potential function on $V_0$ (\cite{Biggs} and \cite{Bollobas}) in which exists under the Kirchhoff current and voltage laws. 
\begin{theorem}[see e.g., \cite{Biggs,Bollobas}]
\label{laplacian}
Let the setting be the same as in Theorem \ref{main2}. Let $G_{0}$ be a bipartite graph and $L$ be the Laplacian matrix of $G_{0}$. Then there exists a potential function $\phi \in \mathbb{C}^{V_{0}}$ such that $j(a)=\phi(o(a))-\phi(t(a))$. Here $\phi$ satisfies the following equation.
$$
L \phi = -q,
$$
where $q(u)=\sum \limits_{a \in \tilde{A}\setminus A_{0}:t(a)=u} j(a),\;u \in V_{0}.$
\end{theorem}
Here we should remark that $q \in (\text{ker}(L))^{\perp}$: actually we have 
\[
0=\sum_{a \in A: t(a) \in V_{0}} j(a) \equiv \sum_{a \in A \setminus A_{0}: t(a) \in V_{0}} j(a)=\sum_{v \in V_{0}} q(v) \bm{1}(v),
\]
where $\bm{1}(\cdot)$ is the constant function whose value is $1$ and $\text{ker}(L)=\left\lbrace c\bm{1} : c \in \mathbb{C} \right\rbrace$.

\quad \\
{\bf Non-bipartite case} \\
When the underlying graph is non-bipartite, 
the following two properties hold for the stationary sate $\Psi_{\infty}$ by Theorem~\ref{thm:pKL}: 
\begin{align}
\Psi_{\infty}(a) &=\Psi_{\infty}(\overline{a}),\;\text{for any}\;a \in \tilde{A}, o(a) \in V_{0}; \label{eq:p1}\\
\sum\limits_{a \in \tilde{A} :t(a)=u} \Psi_{\infty}(a) &= 0, \;\text{for any}\;u \in V_{0}.\label{eq:p2}
\end{align}

As an analogy to the current function in the bipartite case, we represent these properties in term of the non-oriented incidence matrix and the {\it signless} Laplacian matrix in the following theorem. Here the sighless Laplacian matrix of $G_0$ is denoted by $Q=M+D$. 
\begin{theorem}
\label{signlesslaplacian}
Let $\Psi_\infty$ be the stationary state of quantum walk such that $\Psi_\infty(a)=\lim_{n\to\infty}(U^n\Psi_0)(a)$ for any $a\in \tilde{A}$.   
Let $G_{0}$ be a non-bipartite graph and $Q$ be the signless Laplacian matrix of $G_{0}$. Then there uniquely exists $\phi \in \mathbb{C}^{V_{0}}$ such that $\Psi_{\infty}(a)=\phi(o(a))+\phi(t(a))$ for any $a\in A_0$. Here $\phi$ satisfies the following equation.
$$
Q \phi = -q
$$
where $q(u)=\sum \limits_{a \in \tilde{A}\setminus A_{0}:t(a)=u} \Psi_\infty(a),\;u \in V_{0}.$
\end{theorem}

\begin{proof}
Denote the non-oriented incidence matrix on the set of arcs  by $\tilde{C}:\mathbb{C}^{A_{0}} \rightarrow \mathbb{C}^{V_{0}}$ which satisfies

$$(\tilde{C} \psi)(v_{i})= \sum\limits_{j=1}^{m} c_{ij}\psi(a_{j}),\;1 \leq i \leq n,
$$
where $m=|A_{0}|$ and
$$
c_{ij}= 
\begin{cases}
1 & \text{if } \, t(a_{j}) = v_{i} \text{ or } o(a_{j}) = v_{i},\\
0 & \text{otherwise. }
\end{cases}
$$
Then we have
$$(\tilde{C} \psi)(v)= \sum\limits_{a \in A_{0}:t(a)=v} \psi(a)+\sum\limits_{a \in A_{0}:o(a)=v} \psi(a),\;v \in V_{0}.
$$
The adjoint operator $\tilde{C}^{*}:\mathbb{C}^{V_{0}} \rightarrow \mathbb{C}^{A_{0}}$ is given by
$$
(\tilde{C}^{*}f)(a)=f(t(a))+f(o(a)).
$$
Then we have $\tilde{C}\tilde{C}^{*}$ in terms of the signless Laplacian matrix $Q$
$$
\tilde{C}\tilde{C}^{*}=2(D + M)=2Q.
$$
Let $\psi_\infty=\chi^*\Psi_\infty$. 
Now let us see that if there exists a potential function $\phi \in \mathbb{C}^{V_{0}}$ such that $\psi_{\infty}(a)=(\tilde{C}^{*} \phi)(a)$ for any $a\in A_0$, then $\phi$ must satisfy \[ Q\phi=-q. \] 
Remark that  for any $a\in A_0$
$$
\psi_{\infty}(a)=(\tilde{C}^{*} \phi) (a)= \phi (t(a)) + \phi (o(a))=\phi (t(\overline{a})) + \phi (o(\overline{a}))=(\tilde{C}^{*} \phi) (\overline{a})=\psi_{\infty}(\overline{a})
$$
and hence the first property of (\ref{eq:p1}) can be easily confirmed. 
We have, by the property (\ref{eq:p1}),   
\begin{align*}
    \sum \limits_{a \in \tilde{A}:o(a)=u} \Psi_{\infty}(a)&=\sum \limits_{a \in \tilde{A}:o(a)=u} \Psi_{\infty}(\overline{a})\\
    &=\sum \limits_{a \in \tilde{A}:t(a)=u} \Psi_{\infty}(a), 
\end{align*}
and by the property (\ref{eq:p2}),  
$$
\sum \limits_{a \in \tilde{A}:t(a)=u} \Psi_{\infty}(a) + \sum \limits_{a \in \tilde{A}:o(a)=u} \Psi_{\infty}(a)=0, \;u \in V_{0}.
$$
Let us divide the summation in the above equation by
$$
\sum \limits_{a \in A_{0}:t(a)=u} \Psi_{\infty}(a)+\sum \limits_{a \in \tilde{A}\setminus A_{0}:t(a)=u} \Psi_{\infty}(a)+\sum \limits_{a \in A_{0}:o(a)=u} \Psi_{\infty}(a)+\sum \limits_{a \in \tilde{A}\setminus A_{0}:o(a)=u} \Psi_{\infty}(a)=0.
$$
Then it follows that
$$
(\tilde{C}\psi_{\infty})(u)=-2\sum \limits_{a \in \tilde{A}\setminus A_{0}:t(a)=u} \Psi_{\infty}(a).
$$
which implies
$$
Q\phi =\dfrac{1}{2} \tilde{C}\tilde{C}^{*} \phi =-q,
$$
where $q(u)=\sum_{a \in \tilde{A}\setminus A_{0}:t(a)=u} \Psi_{\infty}(a),u \in V_{0}.$ 
Although the existence of the potential function $\phi$ is ensured by the pseudo-Kirchhoff voltage law in Theorem~\ref{thm:pKL}, we prove here directly as follows. 
To show the existence of $\phi$, it is enough to show that
$$
\psi_{\infty} \in \text{Range}(\tilde{C}^{*})=\text{ker}(\tilde{C})^{\perp}.
$$

Let us see that 
\begin{equation}\label{eq:kerC}
\text{ker}(\tilde{C})=(\text{ker}(d) \cap \mathcal{H}_{+}) \oplus 
\mathcal{H}_{-},
\end{equation}
where 
$\mathcal{H}_{+}=\left\lbrace \psi \in \mathbb{C}^{A_{0}} : \psi(a) = \psi(\overline{a}) \text{ for any } a \in A_{0} \right\rbrace$,  $\mathcal{H}_{-}=\left\lbrace \psi \in \mathbb{C}^{A_{0}} : \psi(a) = -\psi(\overline{a}) \text{ for any } a \in A_{0} \right\rbrace$
and 
$$
(d \varphi)(u)=\dfrac{1}{\sqrt{\deg_{G_0}(u)}}\sum \limits_{a \in A_{0}: t(a)=u} \varphi(a).
$$
Note that $\mathcal{H}_{+} \oplus \mathcal{H}_{-}= \mathbb{C}^{A_{0}}$ and $\mathcal{H}_+=\ker(1-S_0)$, $\mathcal{H}_-=\ker(1+S_0)$, where 
$S_0$ is the shift operator on $\mathbb{C}^{A_0}$ such that $S_0\delta_a=\delta_{\bar{a}}$ for any $a\in A_0$. 
The operator $\tilde{C}$ can be rewritten by  
    \[ (\tilde{C}\psi)(u)=\sqrt{\deg_{G_0}(u)} \left( (d\psi)(u)+(dS_0\psi)(u) \right). \]
Then we have
\[ \ker \tilde{C}=\ker(d(1+S_0)), \]
which implies (\ref{eq:kerC}).
By \cite{Higuchi:Segawa} it follows that
$$
\psi_{\infty} \in (\text{ker}(d) \cap \mathcal{H}_{+})^{\perp}. 
$$
Since $\psi_\infty(\bar{a})=\psi_\infty(a)$, then $\psi_\infty\in \mathcal{H}_+=\mathcal{H}_-^\perp$. 
Therefore $\psi_\infty\in \ker(\tilde{C})^\perp=\text{Range}(\tilde{C}^{*})$. 

Furthermore, since $G_{0}$ is non-bipartite, the least eigenvalue of $Q$ is positive (cf. \cite{CRS}) and hence $Q$ is invertible. 
This completes the proof of the existence and the uniqueness of $\phi$.
\end{proof}
$$
$$

\subsubsection{Comfortability }
We have seen that relations between  stationary state and the (signless-)Laplacian, which contains information of the geometry of the graph. 
Let us express the comfortability in terms of some graph geometrical properties. 
This leads the completion of the proof of Theorem~\ref{thm:comfortability}. 
\\

\noindent{\bf Bipartite  case}\\
Now we introduce the energy of the electric circuit $\varEpsilon_{EC}(G_{0})$ which is given by,
$$
\varEpsilon_{EC}(G_{0})=\dfrac{1}{2} \left\lVert j \right\rVert^{2}=\frac{1}{2}\sum_{a\in A_0}|j(a)|^2.
$$

To give the next proposition, we prepare the following notion. Note that the Laplacian matrix $L$ is singular and hence $\phi$ is not determined uniquely. We impose the ground condition $\phi(n)=0$ which reduces the equation in Theorem \ref{laplacian} to $L^{(n)}\phi^{(n)}=-q^{(n)}$, where $L^{(n)}$ is the matrix obtained by removing the $n$-th row and the $n$-th column of the Laplacian matrix, $\phi^{(n)}$ and $q^{(n)}$ are the vectors obtained by removing the $n$-th element from $\phi$ and $q$ respectively. Here, $\text{det}(L^{(n)})$ is the number of spanning trees of $G_{0}$ (see \cite{Biggs}), and hence non-zero. So $\phi^{(n)}$ is determined uniquely by $\phi^{(n)}= -(L^{(n)})^{-1}q^{(n)}$. More over, we denote by $B$ and $\tilde{B}$, the usual incidence matrix and the non-oriented incidence matrix on the set of edges respectively. More precisely, we fix an orientation of each $e\in E_0$ and denote it by $\vec{E}_0$: 
\[A_0=\vec{E}_0\cup \overline{(\vec{E}_0)},\]
where $\overline{(\vec{E}_0)}=\{a\in A_0 \;|\; \bar{a}\in \vec{E}_0\}$. 
Then $B:\mathbb{C}^{\vec{E}_0}\to \mathbb{C}^{V_0}$ is denoted by 
\begin{align*}
    (B\psi)(u) &= \sum_{t(a)=u}\psi(a)-\sum_{o(a)=u}\psi(a);
\end{align*}
then its adjoint is expressed by 
\[  (B^*f)(a) = f(t(a))-f(o(a)).  \]
The non-oriented incidence matrix  $\tilde{B}:\mathbb{C}^{\vec{E}_0}\to \mathbb{C}^{V_0}$ is denoted by 
\begin{align*}
    (\tilde{B}\psi)(u) &= \sum_{t(a)=u}\psi(a)+\sum_{o(a)=u}\psi(a);
\end{align*}
then its adjoint is expressed by 
\[  ({\tilde{B}}^*f)(a) = f(t(a))+f(o(a)).  \]
Now we give the following proposition.

\begin{proposition}
\label{prop}
The electrical energy of the circuit is given by
$$
\varEpsilon_{EC}(G_{0})=\dfrac{1}{\mathrm{det}(L^{(n)})} \sum\limits_{i,j=1}^{n-1} (-1)^{i+j} q^{(n)}(i)q^{(n)}(j) \sum\limits _{\substack{H \subset G_{0}\\|E(H)|=n-2}} \mathrm{det}(B_{H}^{(n,j)})\mathrm{det}((B_{H}^{(n,i)})^{*}),
$$
where $B_{H}^{(n,j)}$ is the matrix obtained by choosing the columns corresponding to the edges in $H$ and removing the $j$-th and $n$-th rows in the oriented incidence matrix of $G_{0}$.
\end{proposition}
\begin{proof}
By definition, we can write the electrical energy in terms of the Laplacian matrix as follows.

\begin{align*}
\varEpsilon_{EC}(G_{0})&=\dfrac{1}{2} \left\lVert j \right\rVert^{2}
=\left\langle B^{*}\phi , B^{*}\phi \right\rangle
=\left\langle \phi , L\phi \right\rangle
=-\left\langle \phi , q \right\rangle
=-\left\langle \phi^{(n)} , q^{(n)} \right\rangle
=\left\langle (L^{(n)})^{-1}q^{(n)} , q^{(n)} \right\rangle\\
&=\sum\limits_{i=1}^{n-1} q^{(n)}(i)((L^{(n)})^{-1}q^{(n)})(i).
\end{align*}
By Cramer's rule, we get,
$$
\varEpsilon_{EC}(G_{0})=\dfrac{1}{\text{det}(L^{(n)})}\sum\limits_{i=1}^{n-1} q^{(n)}(i)\text{det}(L^{(n)}_{i}),
$$
where $\text{det}(L^{(n)}_{i})$ is the matrix obtained by replacing the $i$-th column of $L^{(n)}$ by $q^{(n)}$. Now by expanding along the $i$-th column, we get
$$
\varEpsilon_{EC}(G_{0})=\dfrac{1}{\text{det}(L^{(n)})}\sum\limits_{i,j=1}^{n-1} q^{(n)}(i)q^{(n)}(j)\text{det}(L^{(n)}_{(j,i)}),
$$
where $L^{(n)}_{(j,i)}$ is the matrix obtained by removing the $j$-th row and the $i$-th column from $
L^{(n)}$. Note that $L^{(n)}_{(j,i)}=B^{(j,n)}(B^{(i,n)})^{*}$ and hence
$$
\varEpsilon_{EC}(G_{0})=\dfrac{1}{\text{det}(L^{(n)})}\sum\limits_{i,j=1}^{n-1} q^{(n)}(i)q^{(n)}(j)\text{det}(B^{(j,n)}(B^{(i,n)})^{*}),
$$
where $B^{(j,n)}$ is the matrix obtained by removing the $j$-th and the $n$-th rows from $B$. Now by Binet-Cauchy theorem (see \cite{CRS}), it follows that
$$
\varEpsilon_{EC}(G_{0})=\dfrac{1}{\text{det}(L^{(n)})} \sum\limits_{i,j=1}^{n-1} (-1)^{i+j} q^{(n)}(i)q^{(n)}(j) \sum\limits _{\substack{H \subset G_{0}\\|E(H)|=n-2}} \text{det}(B_{H}^{(n,j)})\text{det}((B_{H}^{(n,i)})^{*}),
$$
where $B_{H}^{(n,j)}$ is the matrix obtained by choosing the columns corresponding to the edges in $H$ from $B^{(j,n)}$.
\end{proof}

Now we give the proof of Theorem \ref{thm:comfortability} by applying  Proposition~\ref{prop}. 
Observing that $\rho=1/2$, $q^{(n)}(1)=1-\rho=1/2$ and $q^{(n)}(i)=0$ $(i=2,\dots,n-1)$ in our setting of Theorem~\ref{thm:comfortability} and 
applying Proposition \ref{prop} to our setting, we get
$$
\varEpsilon_{EC}(G_{0})=\frac{1}{4}\dfrac{1}{\text{det}(L^{(n)})}\sum\limits _{\substack{H \subset G_{0}\\|E(H)|=n-2}} \left(\text{det}(B_{H}^{(n,1)})\right)^{2}.
$$
Focusing on the linearly dependence on the column vectors of the incidence matrix of the spanning subgraph $H\subset G_0$, we obtain the following:
\begin{enumerate}
\item If $H$ contains a cycle, then $ \text{det}(B_{H}^{(n,1)})=0$;
\item If $H$ contains a connected component including both $u_{1}$ and $u_{n}$, then $ \text{det}(B_{H}^{(n,1)})=0$.
\end{enumerate}
Hence, it implies that if $\text{det}(B_{H}^{(n,1)})\neq0$. Then $H$ is a spanning forest which contains exactly two components, one containing $u_{1}$ and the other one containing $u_{n}$. On the other hand, if $H$ is a spanning forest which contains exactly two components, one containing $u_{1}$ and the other one containing $u_{n}$, $B_{H}^{(n,1)}$ is of the form
$$
B_{H}^{(n,1)}=
\left(\begin{array}{@{}c|c@{}}
B_{T_{1}} & 0\\
\hline
0 & B_{T_{2}}
\end{array}\right)
$$
for the trees $T_{1}$ and $T_{2}$ in the forest which are the spanning trees of the two components of the forest. Now by \cite{Biggs}, $\text{det}(B_{T_{1}})=\text{det}(B_{T_{2}})=\pm 1$ and hence $\left(\text{det}(B_{H}^{(n,1)})\right)^2=1$. Thus, $\left(\text{det}(B_{H}^{(n,1)})\right)^2=1$ if and only if $H$ is a spanning forest which contains exactly two components, one containing $u_{1}$ and the other one containing $u_{n}$. Hence it follows that
$$
4 \varEpsilon_{EC}(G_{0})=\dfrac{\chi_{2}(G_{0};u_{1},u_{n})}{\chi_{1}(G_{0})},
$$
where $\chi_{1}(G_{0})$ is the tree number of $G_{0}$ and $\chi_{2}(G_{0};u_{1},u_{n})$ is the number of spanning forests of $G_{0}$ with exactly two components, one containing $u_{1}$ and the other containing $u_{n}$.

Let $G_{0}$ be a bipartite graph. Then
\begin{align}
\varEpsilon_{QW}&=\dfrac{1}{2}\sum \limits _{a \in A_{0}} | \Psi_{\infty}(a) |^{2}
=\dfrac{1}{2}\sum \limits _{a \in A_{0}} | j(a)+ \rho |^{2}
=\dfrac{1}{2} \sum \limits _{a \in A_{0}}  j(a)^{2}+\rho \sum \limits _{a \in A_{0}} j(a)+\rho^{2}|E_{0}| \notag\\
&=\varEpsilon_{EC}(G_{0})+\rho^{2}|E_{0}|.\label{eq:z1OK}
\end{align}
Note that in our setting with only two tails, we have $\rho^{2}=\dfrac{1}{4}$, which leads to the formula in the theorem.\\

\noindent{\bf Non-bipartite case}\\

Now let $G_{0}$ be non-bipartite. Then we have $\psi_{\infty}(e)=\psi_{\infty}(\bar{e})=(\tilde{B}^{*} \phi)(e)$ for any $e\in \vec{E}_0$ and $Q\phi=-q$. Now since $G_{0}$ is non-bipartite, it follows that $Q$ is invertible (for example, see \cite[Theorem 7.8.1]{CRS}) and hence $\phi = -Q^{-1}q$. By a similar argument, it follows that
\begin{align*}
\varEpsilon_{QW}(G_{0})&=\dfrac{1}{2}\sum \limits _{a \in A_{0}} | \Psi_{\infty}(a) |^{2}
=\dfrac{1}{2} \left\langle \psi_{\infty} , \psi_{\infty} \right\rangle
= \left\langle \tilde{B}^{*}\phi , \tilde{B}^{*}\phi \right\rangle
=\left\langle \phi , Q\phi \right\rangle
=\left\langle Q^{-1}q , q \right\rangle.
\end{align*}
By using the Cramer's rule and the Binet-Cauchy theorem, we can derive a similar expression for the non-bipartite graphs as follows:
\begin{proposition}\label{prop:2}
Let $G_0$ be a non-bipartite graph with arbitrary boundary . Then we have 
$$
\varEpsilon_{QW}(G_{0})=\dfrac{1}{\text{det}(Q)} \sum\limits_{i,j=1}^{n} (-1)^{i+j} q(i)q(j) \sum\limits _{\substack{H \subset G_{0}\\|E(H)|=n-1}} \text{det}(\tilde{B}_{H}^{(i)})\text{det}(\tilde{B}_{H}^{(j)}),
$$
where $\tilde{B}_{H}^{(j)}$ is the matrix obtained by choosing the columns corresponding to the edges in $H$ and removing the $j$-th row from $\tilde{B}$. 
\end{proposition}

Observing that $q(j)=\delta_1(j)$ in our setting and applying Proposition~\ref{prop:2}, then we have 
$$
\varEpsilon_{QW}(G_{0})=\dfrac{1}{\text{det}(Q)}\sum\limits _{\substack{H \subset G_{0}\\|E(H)|=n-1}} \left(\text{det}(\tilde{B}_{H}^{(1)})\right)^{2}
.$$
Focusing on the linearly dependence on the column vectors of the incidence matrix of the spanning subgraph $H\subset G\setminus\{{u_1}\}$, we obtain the following:
\begin{enumerate}
\item If $H$ has an even cycle, then $\text{det}(\tilde{B}_{H}^{(1)})=0$;
\item If $H$ has a connected component having at least two odd cycles, then $\text{det}(\tilde{B}_{H}^{(1)})=0$;
\item If $H$ has a connected component having an odd cycle and $u_{1}$, then $\text{det}(\tilde{B}_{H}^{(1)})=0$.
\end{enumerate}
Hence it implies that if $\text{det}(\tilde{B}_{H}^{(1)})\neq 0$ then $H=T \cup C_{1} \cup ... \cup C_{k}$ where $u_{1} \in T$ where $T$ is a tree and $C_{1} \cup ... \cup C_{k}$ are odd unicycles. Now if $H=T \cup C_{1} \cup ... \cup C_{k}$ where $u_{1} \in T$ where $T$ is a tree and $C_{1} \cup ... \cup C_{k}$ are odd unicycles, then $\tilde{B}_{H}^{(1)}$ is of the form
$$
\tilde{B}_{H}^{(1)}=
\left(\begin{array}{@{}c|c|c|c@{}}
\tilde{B}_{T} & 0 & ... & 0\\
\hline
0 & \tilde{B}_{C_{1}} & ... & 0\\
\hline
. & & &\\
. & & &\\
. & & &\\
\hline
0 & ... & 0 & \tilde{B}_{C_{k}}
\end{array}\right).
$$
Note that $\text{det}(\tilde{B}_{T})=\pm 1$ (by \cite{Biggs}) and by expanding the determinant of the non-oriented incidence matrix of a cycle, we can show that $\text{det}(\tilde{B}_{C_{i}})=\pm 2$ (see Appendix for the details), which implies that
$$
\left(\text{det}(\tilde{B}_{H}^{(1)})\right)^2=4^{\omega(H)-1},
$$
where $\omega(H)$ is the number of components in $H$. Furthermore, let
$$
\mathcal{C}_{o}=\{\{C_{1},...,C_{r}\} | C_{i}\text{'s are odd unicyclic graphs}\}
$$
and
$$
\mathcal{T}\mathcal{C}_o=\{\{T,C_{1},...,C_{k}\}| u_{1}\in T \text{ where } T \text{ is a tree and }C_{i}\text{'s are odd unicyclic graphs}\}.
$$
Then by \cite{CRS2007},
$$
\text{det}(Q)=\sum \limits _{H \in C_{0}} 4^{\omega(H)}
$$
and it follows that
$$
\varEpsilon_{QW}=\dfrac{\sum \limits _{H \in \mathcal{T}\mathcal{C}_o} 4^{\omega(H)-1}}{\sum \limits _{H \in C_{0}} 4^{\omega(H)}}
$$
which completes the proof.
$\hfill \square$

\appendix
\section*{Appendix}

Let $C_{i}$ be an odd unicyclic graph and let $\tilde{B}_{C_{i}}$ be the non-oriented incident matrix of $C_{i}$. Let us prove that $\text{det}(\tilde{B}_{C_{i}})=\pm 2$. Since an odd unicyclic graph is obtained by connecting trees to an odd cycle (say $C'$), If we expand $\text{det}(\tilde{B}_{C_{i}})$ along the row corresponding to a leaf of a tree connected to $C'$, it reduces to a similar determinant. By expanding recursively along the rows corresponding to the leaves of the graph, we obtain
$$\text{det}(\tilde{B}_{C_{i}})=\pm \text{det}(B_{C'}).$$

Now it is enough to show that $\text{det}(B_{C'})=2$. Observe that we have
$$
\text{det}(B_{C'})=\text{det}
\left(\begin{array}{@{}c c c c c@{}}
1 & 0 & ... & 0 & 1\\
1 & 1 & 0 & ... & 0\\
. & . & & &\\
. & & . & &\\
. & & & . &\\
0 & ... & 0 & 1 & 1
\end{array}\right).
$$
Expanding the determinant in the right hand side along the first row, since there are odd number of rows and columns, we get
$$
\text{det}(B_{C'})=\text{det}
\left(\begin{array}{@{}c c c c c@{}}
1 & 0 & ... & 0 & 0\\
1 & 1 & 0 & ... & 0\\
. & . & & &\\
. & & . & &\\
. & & & . &\\
0 & ... & 0 & 1 & 1
\end{array}\right) +
\text{det}
\left(\begin{array}{@{}c c c c c@{}}
1 & 1 & ... & 0 & 0\\
0 & 1 & 1 & ... & 0\\
. & . & & &\\
. & & . & &\\
. & & & . &\\
0 & ... & 0 & 0 & 1
\end{array}\right).
$$
Now to compute the two determinants in the right hand side, we expand the first determinant recursively along the rows and the second along the columns and we obtain that these two determinants are equal to $1$. Hence we conclude that 
$$\text{det}(\tilde{B}_{C_{i}})=\pm 2.$$

\vskip\baselineskip
\noindent{\bf Acknowledgments}: Yu.H. acknowledges financial supports from the Grant-in-Aid of
Scientific Research (C) Japan Society for the Promotion of Science (Grant No.~18K03401). 
M.S. acknowledges financial supports from JST SPRING (Grant No.~JPMJSP2114).
E.S. acknowledges financial supports from the Grant-in-Aid of
Scientific Research (C) Japan Society for the Promotion of Science (Grant No.~19K03616) and Research Origin for Dressed Photon. We would like to thank prof. Hajime Tanaka for the invaluable discussions and the comments which were very useful to carry out this study.

\end{document}